\documentclass[a4paper,UKenglish]{lipics-v2018}
\nolinenumbers

\usepackage{microtype}%if unwanted, comment out or use option "draft"

\usepackage{tikz}
\usetikzlibrary{automata, positioning, arrows}

\newtheorem{observation}[theorem]{Observation}
\newtheorem{proposition}[theorem]{Proposition}

\newcommand{\pO}{{\sf Player 1}}
\newcommand{\pT}{{\sf Player 2}}
\newcommand{\AO}{A_1}
\newcommand{\AT}{A_2}

\newcommand{\aO}{a}
\newcommand{\aT}{b}
\newcommand{\Col}{C}
\newcommand{\col}{c}
\newcommand{\tr}{\mathrm{col}}
\newcommand{\hS}{h}
\newcommand{\hAS}{H}
\newcommand{\hOS}{\hS}
\newcommand{\aH}{\rho}
\newcommand{\aHO}{\alpha}
\newcommand{\aHT}{\beta}
\newcommand{\aHinf}{\boldsymbol{\aH}}
\newcommand{\aHOinf}{\boldsymbol{\aHO}}
\newcommand{\aHTinf}{\boldsymbol{\aHT}}
\newcommand{\colH}{\gamma}
\newcommand{\colHinf}{\boldsymbol{\colH}}
\newcommand{\sO}{s}
\newcommand{\sT}{\tau}
\newcommand{\R}{\mathbb{R}}
\newcommand{\N}{\mathbb{N}}
\newcommand{\p}{\mathbb{P}}
\newcommand{\e}{\epsilon}

\title{Concurrent games and semi-random determinacy}
\titlerunning{Concurrent games and semi-random determinacy}%optional, please use if title is longer than one line

\author{St{\'e}phane Le Roux}{Darmstadt Technical University, department of mathematics, Darmstadt, Germany}{leroux@mathematik.tu-darmstadt.de}{}{}%mandatory, please use full name; only 1 author per \author macro; first two parameters are mandatory, other parameters can be empty.

\authorrunning{S. Le Roux} %mandatory. First: Use abbreviated first/middle names. Second (only in severe cases): Use first author plus 'et. al.'

\Copyright{St{\'e}phane Le Roux}%mandatory, please use full first names. LIPIcs license is "CC-BY";  http://creativecommons.org/licenses/by/3.0/

\subjclass{Dummy classification -- please refer to \url{http://www.acm.org/about/class/ccs98-html}}% mandatory: Please choose ACM 1998 classifications from http://www.acm.org/about/class/ccs98-html . E.g., cite as "F.1.1 Models of Computation". 
% mandatory: Please provide 1-5 keywords
\category{}%optional, e.g. invited paper

\relatedversion{}%optional, e.g. full version hosted on arXiv, HAL, or other respository/website

\supplement{}%optional, e.g. related research data, source code, ... hosted on a repository like zenodo, figshare, GitHub, ...

\funding{}%optional, to capture a funding statement, which applies to all authors. Please enter author specific funding statements as fifth argument of the \author macro.

\acknowledgements{The referees of a previous submission made helpful comments. A simplification of the proof of Lemma~\ref{lem:fpi} was triggered by a conversation with Arno Pauly.}%optional

\EventEditors{John Q. Open and Joan R. Access}
\EventNoEds{2}
\EventLongTitle{42nd Conference on Very Important Topics (CVIT 2016)}
\EventShortTitle{CVIT 2016}
\EventAcronym{CVIT}
\EventYear{2016}
\EventDate{December 24--27, 2016}
\EventLocation{Little Whinging, United Kingdom}
\EventLogo{}
\SeriesVolume{42}
\ArticleNo{23}
\begin{document}

\maketitle

\begin{abstract}
Consider concurrent, infinite duration, two-player win/lose games played on graphs. If the winning condition satisfies some simple requirement, the existence of Player 1 winning (finite-memory) strategies is equivalent to the existence of winning (finite-memory) strategies in finitely many derived one-player games. Several classical winning conditions satisfy this simple requirement.

Under an additional requirement on the winning condition, the non-existence of Player 1 winning strategies from all vertices is equivalent to the existence of Player 2 stochastic strategies winning almost surely from all vertices. Only few classical winning conditions satisfy this additional requirement, but a fairness variant of omega-regular languages does.

\end{abstract}

\section{Introduction}

Computer science models systems interacting concurrently with their environment via infinite duration two-player win/lose games played on graphs: a play starts at a \emph{state} of the graph, where the players \emph{concurrently} choose one \emph{action} each and thus induce the next state, and so on for infinitely many rounds. The \emph{winning condition} is a given subset $W$ of the infinite sequences of states, and \pO\/ wins the play iff the sequence of visited states belongs to $W$. A \emph{strategy} of a player prescribes one action depending on what has been played so far, and a \emph{winning strategy} is a strategy ensuring victory regardless of the opponent strategy.

There are games where neither of the players has a winning strategy, but Borel determinacy~\cite{Martin75} guarantees the existence of a winning strategy in games where the players play alternately and the winning condition is a Borel set. Under Borel condition again, Blackwell determinacy~\cite{Martin98} guarantees a weaker conclusion when the players play concurrently: there exists  a value $v \in [0,1]$ such that for all $\e > 0$ the players have stochastic strategies guaranteeing victory with probability $v-\e$ and $1 - v - \e$, respectively.   

In the special case of concurrent games played on finite graphs with $\omega$-regular winning conditions, \cite{AH00} designed algorithms to decide the existence of (stochastic) strategies that are winning, winning with probability one, and winning with probability $1-\e$ for all $\e > 0$. \cite{AH00} also mentions a three-state game where only the latter exist, which exemplifies the complexity of the concurrent $\omega$-regular games on finite graphs. Then \cite{Chatterjee07} studied concurrent prefix independent winning conditions, which is strictly more general than the $\omega$-regular conditions, and \cite{GH10} further improved upon some results. Some of these results were extended recently to multi-player multi-outcome games, see e.g. \cite{BBMU12}, \cite{GMPRW17}.

{\bf The new games}
This article studies slightly different games: when the players concurrently choose one action each, it also produces a \emph{color}; the winning condition is now a given subset $W$ of the infinite sequences of colors; and \pO\/ wins the play iff the produced sequence of colors belongs to $W$. There are two differences between the classical games and the new games. First, the winning condition does not involve the visited states but the transitions instead; second it does so indirectly, via colors labeling the transitions. E.g. in the game on the left-hand side of Figure~\ref{fig:game}, starting at $q_0$, the action sequence $(a_1,b_1)(a_1,b_2)(a_1,b_1)$ yields the state sequence $q_0q_0q_1q_0$ and the color sequence $002$.
\begin{figure}
  \centering
 \[\begin{array}{rl@{\hspace{1cm}}rl@{\hspace{1cm}}c}
  q_0 
  &
\begin{array}{c|c|c|}
   \multicolumn{1}{c}{}&
	  \multicolumn{1}{c}{b_1}&
	  \multicolumn{1}{c}{b_2}\\
	\cline{2-3}
	a_1 & 0,q_0 & 0,q_1\\
	\cline{2-3}
	a_2 & 1,q_0 & 0,q_0\\
	\cline{2-3}	
\end{array}
 &
 q_1
 & 
  \begin{array}{c|c|c|}
   \multicolumn{1}{c}{}&
	  \multicolumn{1}{c}{b_1}&
	  \multicolumn{1}{c}{b_2}\\
	\cline{2-3}
	a_1 & 2,q_0 & 1,q_0\\
	\cline{2-3}
	a_2 & 2,q_1 & 2,q_1\\
	\cline{2-3}	
\end{array}
&
\begin{tikzpicture}
\node[state, initial] (q0) {$q_0$};
\node[right of=q0] (q2) {};
\node[state, right of=q2] (q1) {$q_1$};
\draw (q0) edge[loop above] node{0} (q0)
(q1) edge[->,above] node{1} (q0)
(q1) edge[loop above] node{2} (q1);
\end{tikzpicture}
 \end{array}\] 
  \caption{To the left, a concurrent game with states $q_0$, $q_1$, colors $0,1,2$, and two actions per player. To the right, a one-player game derived by using the delayed response $[(0,q_0)(0,q_0)];[(1,q_0)(2,q_1)]$}\label{fig:game}
\end{figure}
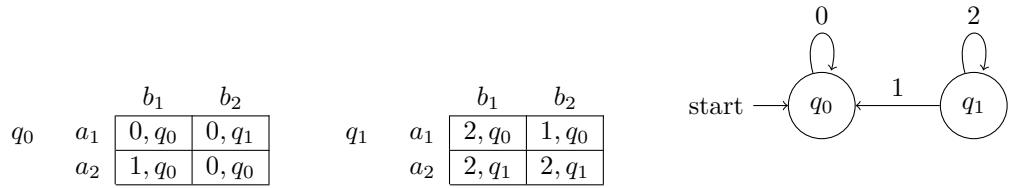

There are several reasons why these new games are interesting.
\begin{itemize}
\item The classical games can be encoded easily into the new ones by using state names as colors. Variants such as the games with colored states, or the colorless games with winning condition on the transitions can also be encoded easily into the new games.

\item The converse encoding may increase the state space (to infinity for games with infinitely many actions). Note that the transition-versus-state issue was already studied in the turned-based setting in \cite{CN06}. Likewise, colorless games are encoded easily in games with colors without size increase, and colors usually lead to more succinct winning conditions.

\item Colors are widely used in turn-based games, and for all games they help to study the winning conditions independently from the game structure, and thus to approximate or even characterize nice winning conditions for classes of games (usually simple to check) rather than for single games (usually more accurate but harder to check). This is exemplified by the difference between Theorems 5 and 7 in \cite{LP18}.

\item Whereas classical one-state games are trivial, the new one-state games are fairly complex and constitute a nice intermediate object towards the understanding of the more complex general games. Likewise, some one-state (aka stateless) objects from the literature are interesting in their own right: \cite{AFKS16} studied one-state multi-objective Markov decision processes; vector addition systems (VAS, \cite{KM69}) are still studied despite the vector addition systems with states (VASS, \cite{HP79}); the Minkowski games \cite{LPR17} defined with finite sets are a special case of the one-state games from this article.   
\end{itemize}

{\bf The main results}
\begin{itemize}
\item If $W$ is closed under \emph{interleaving} and \emph{prefix removal}, and if states and colors are finitely many, the existence of a \pO\/ winning (finite-memory) strategy is equivalent to the existence of winning (finite-memory) strategies in finitely many derived one-player games.

\item If, in addition, $W$ is \emph{factor-prefix complete} and there are finitely many actions, either \pO\/ has a winning strategy from one state, or every \pT\/ constant (stronger than positional!), positive, stochastic strategy is winning almost surely from all vertices. This is \emph{semi-random determinacy}.

\item One-state games enjoy a stronger conclusion under somewhat weaker assumptions: if the winning condition is \emph{factor-set complete} and closed under interleaving, if \pT\/ has finitely many actions, either \pO\/ has a winning strategy, or every \pT\/ constant, positive stochastic strategy is winning almost surely.
\end{itemize}
The finitary flavor of the above characterizations yields decidability and memory sufficiency, in the rough range of double exponentials in the number of states times the number of colors.

In the context of semi-random determinacy, a neutral, random \pT\/ is therefore as bad for \pO\/ as a hostile environment. Also, the victory is clear-cut in the above results: no need for approximate optimal strategies, no need for the notion of value, etc. This is due to the assumptions, and it is legitimate to wonder how restrictive they are.

Several classical winning conditions from computer science are closed under interleaving, see Section~\ref{sect:appli}. The Muller condition is not, but the parity condition is, so the first characterization result extends to the concurrent Muller games via the Last Appearance Record (LAR), as done in \cite{Thomas97}. So, closeness under interleaving is not as restrictive as it may seem.

Fewer classical winning conditions are factor-prefix complete (defined in Section~\ref{subsect:main-results2}), but the boundedness condition from~\cite{LPR17} and a variant of the $\omega$-regular languages are both closed under interleaving and factor-prefix complete. The variant is as follows: each produced color requests some combinations of colors to occur in the future. In winning plays, the number of currently unsatisfied requests should be uniformly bounded over time. It may be relevant even as a business model: at every time unit the system can pay penalties for every currently unsatisfied request, which may be covered by greater, albeit bounded, instantaneous income.

The above variant relates to the notion of \emph{fairness}, which requires that co-finitely many requests are eventually satisfied. The \emph{finitary fairness}~\cite{AH98} additionally requires uniformly bounded response time. This idea was use in \cite{DJP03} to study temporal logic, and in \cite{CHH09} to study finitary parity games. Requiring uniformly bounded response time (or variants thereof) to study games has been further used later, e.g. in \cite{BHR16}. However, these notions of fairness do not enjoy closeness under interleaving and factor-prefix completeness. (Details in Section~\ref{sect:appli}.)

{\bf Related works}
The semi-random determinacy implies the \emph{bounded limit-one property} from \cite{AH00} for the new games: if one state has positive value, one state has value one.

Corollary~\ref{cor:half-pos-one} generalizes the nice Theorem 4 from \cite{Kopczynski06}. Note that the convexity of winning conditions defined in \cite{Kopczynski06} is essentially the same as the interleaving closeness defined here.  

This article also shares similarities with \cite{GZ05}: both use abstract winning conditions, and both characterize the existence of winning strategies in two-player games by the existence of winning strategies in finitely many derived one-player games. Several articles adopted a similar approach: \cite{SLR13} and \cite{SLR14} reduce multi-player multi-outcome Borel games to simpler two-player win/lose Borel games, and characterize the preferences and structures that guarantee the existence of Nash equilibrium in infinite tree-games; \cite{SLR15} does the same to characterize the preferences that guarantee the existence of subgame perfect equilibrium (at low levels of the Borel hierarchy); \cite{LP17} does the same to almost characterize the existence of finite-memory Nash equilibrium in games on finite graphs; \cite{LP14} reduces one-shot concurrent two-player multi-outcome games to simpler one-shot concurrent two-player win/lose games, with applications to generalized Muller games and generalized ``parity'' games. 

One of the benefits of abstraction is that it leads to more general results: e.g. \cite{LP17} noted that the lexicographic product of mean-payoff and reachability objectives cannot be encoded into real-valued payoffs, and \cite{LP18} proved it.

{\bf Structure of the article}
Section~\ref{sect:def} gives basic definitions. Section~\ref{sect:main-results} presents the main results and additional definitions. Section~\ref{sect:proofs} discusses the key elements of the proofs. Section~\ref{sect:appli} presents applications.

\section{Definitions}\label{sect:def}

The folklore Observation~\ref{obs:inf-ext} below will be used extensively to lift properties from finite words to infinite words. It will be first explicitly invoked, and then only implicitly used.

\begin{observation}\label{obs:inf-ext}
Let $f:S^*\to T^*$ be such that $u \sqsubseteq v \Rightarrow f(u) \sqsubseteq f(v)$. Then $f$ can be uniquely extended to $S^* \cup S^\omega \to T^* \cup T^\omega$ such that $f(\rho_{\leq n}) \sqsubseteq f(\rho)$ for all $n \in N$ and $\rho \in S^\omega$.
\end{observation}

{\bf Games}
%\begin{definition}
A \emph{game (with colors and states)} is a tuple $\langle  \AO, \AT, Q, q_0, \delta, \Col, \tr, W\rangle$ such that %the following holds:
\begin{itemize}
\item $\AO$ and $\AT$ are non-empty sets (of actions for \pO\/  and \pT\/),
\item $Q$ is a non-empty set (of states),
\item $q_0 \in Q$ (is the initial state),
\item $\delta: Q \times \AO \times \AT \to Q$ (is the state update function).
\item $\Col$ is a non-empty set (of colors),
\item $\tr : Q \times  \AO \times \AT \to \Col$ (is a color trace),
\item $W \subseteq \Col^\omega$ (is the winning condition for \pO\/)
\end{itemize}
%\end{definition}

{\bf Histories} The \emph{full histories} (\emph{full runs}) of such a game are the finite (infinite) words over $\AO \times \AT$, the \emph{\pT\/ histories} (\emph{\pT\/ runs}) are the finite (infinite) words over $\AT$, and the \emph{\pO\/ histories} (\emph{\pO\/ runs}) are the finite (infinite) words over $\AO$.

{\bf Strategies} A \pO\/ strategy is a function from $\AT^*$ to $\AO$. Informally, it requires \pO\/ to remember exactly how \pT\/ has played so far, and it tells \pO\/ how to play.

{\bf Induced histories} The function $\hOS$ is defined inductively below. As arguments it expects a strategy and a \pT\/ history in $\AT^*$, and it returns a full history: the very full history that, morally, should happen if \pO\/ followed the given strategy while \pT\/ played the given \pT\/ history. $\hOS(\sO,\e) := \e$ and $\hOS(\sO,\aHT\cdot \aT) := \hOS(\sO,\aHT) \cdot (\sO(\aHT),\aT)$.

By Observation~\ref{obs:inf-ext} the function $\hOS$ is extended to expect opponents runs in $\AT^\omega$ and return full runs: $\hOS(\sO,\aHTinf)$ is the only action run whose prefixes are the $\hOS(\sO,\aHTinf_{\leq n})$ for $n \in \N$.

{\bf Winning strategies} A \pO\/ strategy $\sO$ is winning if $\hOS(\sO, \aHTinf) \in W$ for all $\aHTinf \in \AT^\omega$. If there is a \pO\/ winning strategy in a game, one says that \pO\/ wins the game.

{\bf Extending the update and trace functions}
The state update function $\delta$ is extended to $\Delta:(\AO \times \AT)^* \to Q$ inductively: $\Delta(\e) := q_0$ and $\Delta(\aH \cdot (\aO,\aT)) :=\delta(\Delta(\aH), \aO,\aT)$. Using $\Delta$, the trace function $\tr$ is naturally lifted to full histories by induction: $\tr(\e) := \e$ and $\tr(\aH \cdot (\aO,\aT)) := \tr(\aH) \cdot \tr(\Delta(\aH), \aO,\aT)$. The trace is further extended to full runs  by Observation~\ref{obs:inf-ext}. When considering several games, indices may be added to the corresponding $\Delta$ and $\tr$. 

{\bf Memory}
A \pO\/ strategy $\sO$ is said to use memory $M$, or memory size $\log_2|M|$, if there exist a set $M$ and $m_0 \in M$, and two functions $\sigma: Q \times M \to \AO$ and $\mu: Q \times M \times \AT \to M$ such that $\sO(\aHT) = \sigma(\Delta \circ \hOS(\sO,\aHT),m(\aHT))$, where $m$ is defined inductively by $m(\e) := m_0$ and $m(\aHT\aT) := \mu(\Delta\circ \hOS(\sO,\aHT),m(\aHT),\aT)$. If $M$ is finite, $\sO$ is called a finite-memory strategy. Note that every \pO\/ strategy uses memory $\AT^\omega$.

{\bf One-player games}
Intuitively, a \emph{one-player game (with colors and states)} amounts to a game where \pT\/ has only one strategy available, i.e. $|\AT| = 1$. Formally, it is a tuple $\langle  \AO, Q, q_0, \delta, \Col, \tr, W\rangle$ such that $\AO$, $Q$, and $\Col$ are non-empty sets, $q_0 \in Q$, $\delta: Q \times \AO \to Q$, $\tr : Q \times \AO \to \Col$, and $W \subseteq \Col^\omega$. In this context, the \emph{full histories} (\emph{full runs}) of such a game are the finite (infinite) words over $\AO$, and the \emph{\pT\/ histories} of \pO\/ are the natural numbers (telling how many rounds have been played). There is only one \emph{\pT\/ run}, namely $\omega$. Then, a \pO\/ strategy is a function from $\N$ to $\AO$, and the notation for the induced full histories is overloaded: $\hOS(\sO,0) := \e$ and $\hOS(\sO,n+1) := \hOS(\sO,n) \cdot \sO(n)$. By Observation~\ref{obs:inf-ext} the function $\hOS$ is (again) extended: $\hOS(\sO,\omega)$ is the only action run whose prefixes are the $\hOS(\sO,n)$ for $n \in \N$. A \pO\/ strategy $\sO$ is winning if $\tr \circ \hOS(\sO,\omega) \in W$.

%{\bf Game restriction}
%Consider a game $g = \langle  \AO, \AT, Q, q_0, \delta, \Col, \tr, W\rangle$. Let $Q_0 \subseteq Q$ be such that  $q_0 \in Q_0$ and \pO\/ \emph{can ensure $Q_0$ uniformly}, i.e. \pO\/ can ensure that, starting from any $q \in Q_0$, only states in $Q_0$ are visited. For all $q\in Q_0$ let $\AO^q := \{\aO \in \AO\,\mid\, \delta(q,\aO,\AT) \subseteq Q_0\}$. By assumption $\AO^q \neq \emptyset$, so let $\aO^q \in \AO^q$. The restriction of $g$ to $Q_0$ by \pO\/ is defined by $g\mid_{\pO\/ : Q_0} := \langle \AO, \AT, Q_0, q_0, \delta_r, \Col, \tr_r, W\rangle$, where $(\tr_r,\delta_r)(q,\aO,\aT) := (\tr,\delta)(q,\aO,\aT)$ if $\aO \in \AO^q$, and $(\tr,\delta)(q,\aO^q,\aT)$ otherwise. (Note that defining a restriction by deleting the actions $\AO \setminus \AO^q$ at each state $q \in Q_0$ is not possible since the current formalism  requires that the available actions be the same at each state.)

{\bf Prefix removal} A set of infinite sequences is closed under prefix removal if the tails of the sequences from the set are again in the set. Formally, $W \subseteq \Col^\omega$ is closed under prefix removal if the following holds: $\forall (\colH,\colHinf) \in \Col^* \times \Col^\omega,\,\colH \cdot \colHinf \in W\,\Rightarrow\, \colHinf \in W$. Note that closeness under prefix removal is weaker than the prefix independence assumed in~\cite{Chatterjee07}, \cite{GH10}, and \cite{Kopczynski06}.

{\bf Interleaving} Interleaving two infinite sequences consists in enumerating sequentially (part of) the two sequences to produce a new infinite sequence. For example, interleaving $(2n)_{n\in \N}$ and $(2n+1)_{n\in \N}$ can produce the sequences $(n)_{n\in \N}$ (perfect alternation), $1 \cdot 0 \cdot 3 \cdot 5 \cdot 2 \cdot 7 \cdot 4 \cdot 6 \cdot (n+8)_{n\in \N}$, and $(2n)_{n\in \N}$ (by enumerating the first sequence only), but not the sequences $(4n)_{n\in \N}$ or $0 \cdot 1 \cdot 4 \cdot 3 \dots$.

{\bf Delayed response}
Consider a game $g = \langle  \AO, \AT, Q, q_0, \delta, \Col, \tr, W\rangle$ with finite $Q$ and $\Col$. For every $q \in Q$ let $E_1^q,\dots,E_{k_q}^q$ be the elements of $\{(\tr,\delta)(q,\aO,\AT)\,\mid\, \aO \in \AO\}$, where $(\tr,\delta)(q,\aO,\AT) := \{(\tr(q,\aO,\aT),\delta(q,\aO,\aT)) \,\mid\, \aT \in \AT\}$ for all $\aO \in \AO$. The elements of  $\otimes_{q\in Q, i \leq k_q}E_i^q$ are called the \pT\/ \emph{delayed responses}. Intuitively, a \pT\/ delayed response amounts to a \pT\/ positional strategy in (and only in) a sequentialized version of the game. In every round of this version, \pO\/ chooses an action first, then \pT\/ chooses an action (or more precisely some color and state among the pairs he could induce by choosing an action). E.g. $[(0,q_0)(0,q_0)];[(1,q_0)(2,q_1)]$ is a delayed response for Figure~\ref{fig:game}. It means that at state $q_0$, \pT\/ selects $(0,q_0)$ for both actions of \pO\/, and at state $q_1$ it selects $(1,q_1)$ if \pO\/ chooses action $a_1$. Note that delayed responses are not \pT\/ (positional) strategies in the concurrent game, e.g. as $[(0,q_0)(0,q_0)]$ is not achievable in any column.

{\bf Derived one-player games}
Let $t$ be  a \pT\/ delayed response. The  one-player game $g(t) := \langle  \AO, Q, q_0, \delta_t, \Col, \tr_t, W\rangle$ is defined by $(\tr_t,\delta_t)(q,\aO) := t_{q,(\tr,\delta)(q,\aO,\AT)}$, the projection of $t$ on the $(q,E_i^q)$-component such that $E_i^q = (\tr,\delta)(q,\aO,\AT)$. Intuitively, $g(t)$ is the game obtained by letting \pT\/ fix his strategy (to realize) $t$ in the sequentialized version of $g$. For example, the game on the left-hand side of Figure~\ref{fig:game} applied to the delayed response $[(0,q_0)(0,q_0)];[(1,q_0)(2,q_1)]$ yields the game on the right-hand side of Figure~\ref{fig:game}.

\section{Main results}\label{sect:main-results}

Section~\ref{subsect:main-results1} characterizes the existence of \pO\/ winning strategies and gives a complexity result.
Section~\ref{subsect:main-results2} defines additional concepts and uses the above characterization to characterize  the existence of \pT\/ everywhere-winning stochastic strategies.
Section~\ref{subsect:main-results3} studies the special case of one-state games and presents the semi-random determinacy.

\subsection{Existence of \pO\/ winning strategies}\label{subsect:main-results1}

Theorem~\ref{thm:interleaving-prefix removal-eq} below characterizes the existence of \pO\/ winning strategies in a game via the existence of winning strategies in finitely many derived one-player games. Theorem~\ref{thm:interleaving-win-all-eq} afterwards drops the assumption on closeness under prefix removal from Theorem~\ref{thm:interleaving-prefix removal-eq},  but at the cost of a universal quantification over the starting state of the game. In Theorems~\ref{thm:interleaving-prefix removal-eq} and \ref{thm:interleaving-win-all-eq}, the finiteness and the closeness  assumptions are used only to prove the \ref{thm:interleaving-win-all-eq2}$\,\Rightarrow\,$\ref{thm:interleaving-win-all-eq1} implications.
\begin{theorem}\label{thm:interleaving-prefix removal-eq}
Consider a game $g = \langle  \AO, \AT, Q, q_0, \delta, \Col, \tr, W\rangle$. If $Q$ and $\Col$ are finite, and $W$ is closed under interleaving and prefix removal, the following are equivalent.
\begin{enumerate}
\item\label{thm:interleaving-prefix removal-eq1} \pO\/ wins $g$.
\item\label{thm:interleaving-prefix removal-eq2} \pO\/ wins $g(t)$ for all delayed responses $t$.
\end{enumerate}
If $\AO$ is finite and \pO\/ wins, she can do it with memory size $O(f(|\AO|,|Q|,|\Col|)  \cdot (|\Col \times Q|)^{|Q|2^{|\Col \times Q|}})$, where $f(|\AO|,|Q|,|\Col|)$ is a sufficient memory size to win the one-player games using $\AO$, $Q$ and $\Col$.
\end{theorem}
\begin{theorem}\label{thm:interleaving-win-all-eq}
Consider games $g_{q} = \langle  \AO, \AT, Q, q, \delta, \Col, \tr, W\rangle$ parametrized by $q \in Q$. If $Q$ and $\Col$ are finite and $W$ is interleaving-closed, the following are equivalent.
\begin{enumerate}
\item\label{thm:interleaving-win-all-eq1} \pO\/ wins $g_{q}$ for all $q \in Q$.
\item\label{thm:interleaving-win-all-eq2} \pO\/ wins $g_{q}(t)$ for all $q \in Q$ and delayed responses $t$.
\end{enumerate}
If the above holds, \pO\/ wins every $g_q$ with memory size as in Theorem~\ref{thm:interleaving-prefix removal-eq}.
\end{theorem}

In games that are (or encode) turn-based games, the delayed responses are \pT\/ positional strategies. So, restricting Theorems~\ref{thm:interleaving-prefix removal-eq} and \ref {thm:interleaving-win-all-eq} to turn-based games yields Corollaries~\ref{cor:half-pos-one} and \ref{cor:half-pos-all}, respectively. Note that Corollary~\ref{cor:half-pos-one} generalizes Theorem 4 from \cite{Kopczynski06} by only assuming closeness under prefix removal instead of prefix independence. This is significant since the safety condition is closed under interleaving and prefix removal, but is not prefix independent.
\begin{corollary}\label{cor:half-pos-one}
Consider a game $g = \langle  \AO, \AT, Q, q_0, \delta, \Col, \tr, W\rangle$ encoding a turn-based game. If $Q$ and $\Col$ are finite, and $W$ is closed under interleaving and prefix removal, either \pO\/ has a winning strategy or \pT\/ has a positional winning strategy.
\end{corollary}
\begin{corollary}\label{cor:half-pos-all}
Consider games $g_{q} = \langle  \AO, \AT, Q, q, \delta, \Col, \tr, W\rangle$ parametrized by $q \in Q$ and encoding a turn-based games. If $Q$ and $\Col$ are finite, and $W$ is closed under interleaving, either \pO\/ wins all $g_q$, or \pT\/ has a positional winning strategy for some $g_q$.
\end{corollary}

The characterizations from Theorems~\ref{thm:interleaving-prefix removal-eq} and \ref {thm:interleaving-win-all-eq} yields decidability results and rough algorithmic complexity estimates in Corollary~\ref{cor:interleaving-win-one-eq} below. Note that checking all the possible strategies using memory size given by Theorems~\ref{thm:interleaving-prefix removal-eq} and \ref {thm:interleaving-win-all-eq} would be slower than Corollary~\ref{cor:interleaving-win-one-eq}.
\begin{corollary}\label{cor:interleaving-win-one-eq}
Let $\mathcal{\Col} \neq \emptyset$, let $W \subseteq \mathcal{\Col}^\omega$ be closed under interleaving  and prefix removal (resp. by interleaving), and let $f: \N^3 \to \N$ be such that for all finite $\Col \subseteq \mathcal{\Col}$ and all one-player games $\langle  \AO, Q, q_0, \delta, \Col, \tr, W\rangle$, it takes at most $f(|\AO|, |Q|, |\Col|)$ computation steps to decide the existence of a (finite-memory) winning strategy in the game. Then for all finite games $g_{q_0} = \langle \AO, \AT, Q, q_0, \delta, \Col, \tr, W\rangle$ it takes at most 
$$
f(|\AO|,|Q|,|\Col|) \cdot (|\Col \times Q|)^{|Q|2^{|\Col \times Q|}} +|Q||\AO||\AT| 
$$
computation steps to decide whether \pO\/ wins $g_{q_0}$ (with finite memory).

(resp. $|Q| \cdot f(|\AO|,|Q|,|\Col|) \cdot (|\Col \times Q|)^{|Q|2^{|\Col \times Q|}} +|Q||\AO||\AT|$ computation steps to decide whether \pO\/ wins $g_{q}$ (with finite memory) for all $q \in Q$.)
\end{corollary}

\subsection{Existence of \pT\/ almost surely winning random strategies}\label{subsect:main-results2}

Consider a game $\langle  \AO, \AT, Q, q_0, \delta, \Col, \tr, W\rangle$.

{\bf Probability distribution} A probability distribution on a finite set $E$ is a function $f: E \to [0,1]$ such that $\sum_{e \in E}f(e) = 1$. Let us call $D(E)$ the set of the probability distributions on $E$.

{\bf Stochastic strategies}
A \pT\/ stochastic strategy is a function $\sT: (\AO \times \AT)^* \to D(\AT)$.

{\bf Induced stochastic histories} The function $\hAS$ is defined inductively below. As arguments it expects a \pT\/ stochastic strategy $\sT$ and a \pO\/ history $\aHO \in \AO^*$, and it returns a probability distribution on $(\AO \times \AT)^{|\aHO|}$. Informally, it tells the probability of a full history if \pO\/ plays $\aHO$ and \pT\/ follows $\sT$ for $|\aHO|$ many rounds. Formally, $\hAS(\sT, \e)(\e) := 1$, and $\hAS(\sT,\aHO \cdot \aO) (\aH \cdot (\aO',\aT)) := 0$ if $\aO' \neq \aO$, and $\hAS(\sT,\aHO \cdot \aO) (\aH \cdot (\aO,\aT)) := \hAS(\sT,\aHO)(\aH)  \cdot \sT(\aH)(\aT)$.

{\bf Induced probability measure} The function $\hAS$ is extended to expect \pO\/ runs: $\forall \aHOinf \in \AO^\omega,\,\hAS(\sT,\aHOinf) (\aH) :=\hAS(\sT,\aHOinf_{<|\aH|}) (\aH)$. Thus, every pair $(\sT,\aHOinf) \in D(\AT)^{(\AO \times \AT)^*} \times \AO^\omega$ induces a probability measure $\lambda(\sT,\aHOinf)$ on $(\AO \times \AT)^\omega$.

{\bf Almost surely winning stochastic strategies}
A \pT\/ stochastic strategy $\sT$ is said to be winning almost surely if $\lambda(\sT,\aHOinf)(\tr^{-1}(W)) = 0$ for all $\aHOinf \in \AO^\omega$.

{\bf Factor-prefix completeness}
Informally, $W$ is factor-prefix complete if the following holds: if the prefixes of an infinite sequence occur as factors arbitrarily far in the tail of a second sequence in $W$, the first sequence is also in $W$. (A factor, aka substring, is a subsequence of consecutive elements.) 
Formally, $W \subseteq \Col^\omega$ is factor-prefix complete if the following holds: $\forall \colHinf \in \Col^\omega,(\exists \colHinf’ \in W, \forall n,m \in \N,\exists k \in \N,\colHinf_{\leq n} = \colHinf’_{m+k} \dots \colHinf’_{m+k+n})\Rightarrow\colHinf \in W$.

%Formally, $W \subseteq \Col^\omega$ is factor-prefix complete if the following holds: for all $\rho \in W$, for all monotone $f,g : \N \to \N$ such that $\forall n \in N, f(n) + g(n) < f(n+1)$, there exist $\rho' \in W$ and a monotone $\varphi: \N \to \N$ such that $\rho_{f \circ \varphi(n)} \dots  \rho_{f \circ \varphi(n) + g \circ \varphi(n)}\sqsubseteq \rho'$ for all $n \in \N$.

In Theorem~\ref{thm:no-win-all-random-win} below, a distribution is said to be positive if it assigns only positive masses. A (stochastic) strategy is said to be constant if it is a constant function, i.e. it returns always the same distribution, which is stronger than being Markovian, memoriless, or positional.
\begin{theorem}[semi-random determinacy]\label{thm:no-win-all-random-win}
Consider games $g_{q} = \langle  \AO, \AT, Q, q, \delta, \Col, \tr, W\rangle$ parametrized by $q \in Q$. If $\AO$ and $\AT$ are finite, if $W$ is factor-prefix complete and closed under interleaving and prefix removal, the following are equivalent.
\begin{enumerate}
\item\label{thm:no-win-all-random-win1} for all $q \in Q$, \pO\/ has no winning strategy in $g_q$.
\item\label{thm:no-win-all-random-win2} for all $q \in Q$, \pT\/ has a constant, positive, stochastic strategy winning $g_q$ almost surely.
\item\label{thm:no-win-all-random-win3} for all $q \in Q$ every \pT\/ stochastic strategy involving probabilities bounded away from $0$ (i.e. with positive infimum) wins $g_q$ almost surely.
\end{enumerate}
\end{theorem}
So in the setting of Theorem~\ref{thm:no-win-all-random-win}, either \pO\/ has a winning strategy for some $g_q$, or every constant, positive strategy is winning almost surely, hence the determinacy. Also note that semi-random determinacy implies the \emph{bounded limit-one property} from \cite{AH00} for the new games: if one state has positive value, one state has value one.

\subsection{The special case of stateless (i.e. one-state) games}\label{subsect:main-results3}

{\bf Stateless games}
Intuitively, a \emph{stateless game (with colors)} amounts to a game with only one state, i.e. $|Q| = 1$. Formally, it is a tuple $\langle  \AO,\AT, \Col, \tr, W\rangle$ such that $\AO$, $\AT$, and $\Col$ are non-empty sets, $\tr : \AO \times \AT \to \Col$ (as opposed to $\tr : Q \times \AO \times \AT \to \Col$ in the general case), and $W \subseteq \Col^\omega$. Histories, runs, strategies, and induced histories are defined as in the general case . It is easier to extend the trace in this context: $\tr(\e) := \e$ and $\tr(\aH \cdot (\aO,\aT)) =\tr(\aH) \cdot \tr(\aO,\aT)$.

Restricting Theorem~\ref{thm:interleaving-win-all-eq} to stateless games yields a simpler Corollary~\ref{cor:eq-win-abs} below. (Note that restricting Theorem~\ref{thm:interleaving-prefix removal-eq} would yield a weaker variant of Corollary~\ref{cor:eq-win-abs}, i.e. additionally assuming closeness under prefix removal.) Memory size and algorithmic complexity estimates could be obtained essentially by replacing $|Q|$ with $1$ in Theorem~\ref{thm:interleaving-win-all-eq} and Corollary~\ref{cor:interleaving-win-one-eq}.
\begin{corollary}\label{cor:eq-win-abs}
Consider a game $\langle \AO, \AT, \Col, \tr, W\rangle$ with finite $\Col$ and interleaving-closed $W$. Let $\Col_1,\dots,\Col_k$ be the elements of $\{\tr(\aO,\AT)\,\mid\, \aO \in \AO\}$. The following are equivalents.
\begin{enumerate}
\item\label{cor:eq-win-abs1} \pO\/ has a winning strategy (resp. finite-memory winning strategy).
\item\label{cor:eq-win-abs2} $\forall (\col_1,\dots,\col_k) \in \Col_1 \times \dots \times \Col_k,\, W \cap \{\col_1,\dots,\col_k\}^\omega \neq \emptyset$ (resp. $W \cap \{\col_1,\dots,\col_k\}^\omega \cap \mathrm{reg}_\Col \neq \emptyset$),
\end{enumerate}
where $\mathrm{reg}_\Col$ are the regular infinite sequences over $\Col$.
\end{corollary}

Restricting Theorem~\ref{thm:no-win-all-random-win} to stateless games cancels the universal quantification over states, but an even stronger version can be obtained: finiteness of $\AO$ and prefix removal closeness are dropped, and the assumption on factor-prefix completeness is weaken into factor-set completeness, as below.

{\bf Factor-set completeness}
A language of infinite sequences is called \emph{factor-set complete} if the following holds: if a sequence in the language has factors of unbounded length over some $\Col_0$, the language has a sequence over $\Col_0$. This is formally defined by contraposition: $W \subseteq \Col^\omega$ is factor-set complete if for all $\Col_0 \subseteq \Col$ and for all $\rho \in W$, we have $W \cap \Col_0^\omega = \emptyset \,\Rightarrow\,\forall\rho \in W,\exists m \in N, \forall n \in N, \exists i \in N, i < m \wedge \rho_{n+i} \notin \Col_0$.

\begin{observation}\label{obs:factor-pref-set}
Factor-prefix completeness implies factor-set completeness (finite alphabets).
\end{observation}

\begin{theorem}[Stateless semi-random determinacy]\label{thm:stateless-semi-rand-det}
Consider a stateless game $\langle \AO, \AT, \Col, \tr, W\rangle$ with finite $\Col$ and $\AT$. Let us assume that $W$ is interleaving-closed and factor-set complete. Then either \pO\/ has a winning strategy, or every \pT\/ constant, positive, stochastic strategy is winning almost surely.
\end{theorem}

\section{The proofs}\label{sect:proofs}

Theorems~\ref{thm:interleaving-prefix removal-eq} and \ref {thm:interleaving-win-all-eq} characterize a concurrent game by finitely many one-player games. A natural idea would be to split their proof into two parts: first, reduce the problem to turn-based games via the well-known observation that a player has a winning strategy in a concurrent game iff she has one in the sequential version of the game where she plays first; second, use similar techniques as in \cite{Kopczynski06}. For this to work, the sequential versions of the concurrent games must allow for colorless transitions, or a fresh color should be used for the transitions where \pO\/ plays. This raises three issues: first, true colors should occur infinitely often in every run in these turn-based games, which would require a more complex notion of turn-based game; second, the winning condition should be rephrased to take the fresh color into account, and so should its closeness properties; third, it would be much difficult to obtain stronger results for the one-state concurrent games, since the one-state property may be hard to track through the translation into turn-based games. Instead, this article overcomes the concurrency directly thanks to Lemma~\ref{lem:fpi}.

\begin{lemma}\label{lem:fpi}
Let $(X_i)_{i \in I}$ be a family of sets. Then

$\forall f : \prod_{i \in I}X_i  \to I, \exists i \in I, \forall x \in X_i, \exists y \in \prod_{i \in I}X_i, y_i = x \wedge f(y) = i$.
\begin{proof}
Towards a contradiction, let us assume the negation of the claim, i.e. $\exists f : \prod_{i \in I}X_i  \to I, \forall i \in I, \exists x \in X_i, \forall y \in \prod_{i \in I}X_i, y_i \neq x \vee f(y) \neq i$. By collecting one witness $x =: z_i$ for each $i$, one constructs $z \in \prod_{i \in I}X_i$ such that $\forall y \in \prod_{i \in I}X_i, y_i \neq z_i \vee f(z) \neq i$. In particular, taking $y := z$ yields $z_i \neq z_i \vee f(z) \neq i$ for all $i$, which contradicts the type of $f$.
\end{proof}
\end{lemma}

Consider the one-state game $g$ in Figure~\ref{fig:one-state-game} (to the left), where each cell encloses one vector of the real plane. \pO\/'s objective is that the sum of the outcome vectors remains bounded, which is closed under interleaving and prefix removal, so $g$ is a concurrent version of the Minkowski games \cite{LPR17}. There are $2^3 = 8$ delayed responses, and five of the corresponding one-player games $g_0,\dots g_7$ are displayed to the right in Figure~\ref{fig:one-state-game}. \pO\/ wins $g_0,\dots,g_7$, since for each $i \leq 7$ the vector $(0,0)$ is in the convex hull of the three vectors defining $g_j$. The idea is to let \pO\/ play $g$ as if she were playing $g_0,\dots,g_7$ in parallel, more specifically in an interleaved way. Then, summing up the eight bounded trajectories yields a bounded trajectory for $g$.

The main difficulty to play the $g_0,\dots,g_7$ in an interleaved way is that at every stage, \pO\/ should pick an action such that whichever action \pT\/ chooses, the resulting vector is exactly the expected one by the (fixed) winning strategy for \emph{some} $g_j$. Let $f:\{1,2\}^3 \to \{a_1,\dots,a_3\}$ be the function that tells which action should be played currently in each of the $2^3 = 8$ one-player games. By Lemma~\ref{lem:fpi} there exists an action $a_i$ such that the following holds: if \pT\/ chooses $b_1$, there exists $g_j$ expecting the vector in the cell $(a_i,b_1)$, and likewise if \pT\/ chooses $b_2$, there exists $g_k$ expecting the vector in the cell $(a_i,b_2)$.
\begin{figure}
\centering
\[
\begin{array}{c@{\hspace{1cm}}ccccc}
\begin{array}{c|c|c|}
   \multicolumn{1}{c}{}&
	  \multicolumn{1}{c}{b_1}&
	  \multicolumn{1}{c}{b_2}\\
	\cline{2-3}
	a_1 & 1,1 & 2,1\\
	\cline{2-3}
	a_2 & 1,-1 & 0,-1\\
	\cline{2-3}
	a_3 & -1,0 & -2,0\\
	\cline{2-3}	
\end{array}
&
\begin{array}{|c|}
   \multicolumn{1}{c}{}\\
	\hline
	 1,1\\
	\hline
	 1,-1\\
	\hline
	-1,0\\
	\hline	
\end{array}
&
\begin{array}{|c|}
   \multicolumn{1}{c}{}\\
	\hline
	 1,1\\
	\hline
	 1,-1\\
	\hline
	-2,0\\
	\hline	
\end{array}
&
\begin{array}{|c|}
   \multicolumn{1}{c}{}\\
	\hline
	 1,1\\
	\hline
	 0,-1\\
	\hline
	-1,0\\
	\hline	
\end{array}
&
\begin{array}{|c|}
   \multicolumn{1}{c}{}\\
	\hline
	 1,1\\
	\hline
	 0,-1\\
	\hline
	-2,0\\
	\hline	
\end{array}
&
\begin{array}{|c|}
   \multicolumn{1}{c}{}\\
	\hline
	 2,1\\
	\hline
	 1,-1\\
	\hline
	-1,0\\
	\hline	
\end{array}
\end{array}
\]
  \caption{A concurrent Minkowski game and its derived games}\label{fig:one-state-game}
\end{figure}

Let us now quickly mention semi-random determinacy. The proof of Theorem~\ref{thm:no-win-all-random-win} below uses similar techniques as, e.g., a proof in the submitted journal version of \cite{LPR17}.
\begin{proof}[Proof of \ref{thm:no-win-all-random-win1} $\Rightarrow$ \ref{thm:no-win-all-random-win3} from Theorem~\ref{thm:no-win-all-random-win}]
Let $p \in ]0,\frac{1}{|\AT|}]$ and let $\sT$ be a \pT\/ stochastic strategy that always assigns probability at least $p$ to every action. 

For all $q \in Q$, by contraposition of Theorem~\ref{thm:interleaving-prefix removal-eq} let $t_q$ be a delayed response (in $g_q$) such that \pO\/ loses the one-player game $g_q(t_q)$. For all $n \in \N$, anytime a play reaches the state $q$, the probability that from then on \pT\/ follows $t_q$ for $n$ rounds in a row, as if second-guessing \pO\/, is  greater than or equal to $p^n$.

Consider a play where \pT\/ follows $\sT$. Let $q$ be a state that is visited infinitely often. (Such a state exists since $Q$ is finite.) Thanks to the argument above, for all $n \in \N$, the probability that, at some point, \pT\/ follows $t_q$ for $n$ rounds in a row from $q$ on is one. Since the countable intersection of measure-one sets has also measure one, the probability that, for all $n \in \N$, at some point \pT\/ follows $t_q$ for $n$ rounds in a row from $q$ on is one.

Let $(\aH^n)_{n \in \N}$ be the corresponding full histories. Since $\AO$ and $\AT$ are finite, the tree induced by prefix closure of the $(\aH^n)_{n \in \N}$ is finitely branching, so by Koenig's Lemma it has an infinite path $\aHinf$, which corresponds to \pT\/ following $t_q$ infinitely many rounds in a row. So $\tr(\aHinf) \notin W$. By factor-prefix closeness the original play is also losing for \pO\/, i.e. winning for \pT\/.
\end{proof}

\section{Applications}\label{sect:appli}

Abstract assumptions need not only be general, they also need to be practical. Section~\ref{sect:apcc} shows that the closeness and completeness axioms enjoy nice algebraic properties: individually, w.r.t. Boolean combination, as well as collectively via the derived closure or completion operators. Section~\ref{sect:cwc} mentions several classical or recent winning conditions from computer science and tells which of them satisfy the closeness and completeness axioms. Section~\ref{sect:brl} introduces the notion of bounded residual load as an alternative to the finitary fairness~\cite{AH98}, and uses it to define a finitary variant of the $\omega$-regular languages that satisfies the closeness and completeness axioms.

\subsection{Algebraic properties of the closeness and completeness axioms}\label{sect:apcc}

Lemma~\ref{lem:bool-combo} below shows how the axioms behave w.r.t. Boolean combination.
\begin{lemma}\label{lem:bool-combo}
\begin{enumerate}
\item\label{lem:bool-combo1} The set of the factor-set complete languages is closed under union.

\item\label{lem:bool-combo2} The set of the interleaving-closed languages is closed under intersection.

\item\label{lem:bool-combo3} The set of the factor-prefix complete languages is closed under intersection and union.
\end{enumerate}
\end{lemma}

The set of the interleaving-closed languages is not closed under union: $\{0^\omega\}$ and $\{1^\omega\}$ are closed under interleaving (and by prefix removal), but $\{0^\omega,1^\omega\}$ is not. The set of the interleaving-closed languages is not closed under complementation: the interleaving of two infinite sequences that are not eventually constant is not eventually constant, but interleaving the eventually constant sequences $0^\omega$ and $1^\omega$ may yield $(01)^\omega$. The set of the factor-set complete languages is not closed under intersection: indeed, both two-element sets $\{0(12)0(12)^20(12)^30\dots, (12)^\omega\}$ and $\{0(12)0(12)^20\dots, (112)^\omega\}$ are factor-set complete, but their intersection $\{0(12)0(12)^20\dots\}$ is not. The set of the factor-set (-prefix) complete languages is not closed under complementation: $\{1^\omega\}$ is factor-set (-prefix) complete, but $\{0,1\}^\omega \setminus \{1^\omega\}$ is not.

The closeness under interleaving and prefix removal, and the factor-prefix completeness induce closure operators. If a relevant winning condition fails to satisfy an equaly relevant axiom, such an operator conveniently constructs a (more generous, axiom satisfying) variant of the winning condition. The closure by prefix removal of a set consists in adding the tails of the sequences from the set; the closure by interleaving consists in adding sequences obtained by interleaving the sequences from the set; and the factor-prefix completion consists in adding the sequences whose prefixes occur arbitrarily far in a sequence from the set. Note that factor-set completeness does not induce a canonical closure operator due to the existential quantifier in its definition.

Lemma~\ref{lem:closure-closed} below shows that the operators behave as expected. This is not for granted in general, as one may need to perform the addition operation an ordinal number of times. Here, one step suffices, which is convenient if computation is of concern.
\begin{lemma}\label{lem:closure-closed}
\begin{enumerate}
\item Closure by prefix removal yields sets that are closed under prefix removal.

\item Closure by interleaving yields sets that are closed under interleaving

\item Factor-prefix completion yields sets that are factor-prefix complete.
\end{enumerate}
\end{lemma}

Lemma~\ref{lem:combine-closure} shows that the operators preserve the existing properties.(Lemma~\ref{lem:closure-closed} is invoked as a proof technique.)
\begin{lemma}\label{lem:combine-closure}
\begin{enumerate}
\item Closure by prefix removal preserves closeness under interleaving.

\item Closure by prefix removal preserves factor-set and factor-prefix completeness.

\item Closure by interleaving preserves closeness under prefix removal.

\item Closure by interleaving preserves factor-set and factor-prefix completeness.

\item Factor-prefix completion preserves closeness under prefix removal.
\end{enumerate}
\end{lemma}

\subsection{Concrete winning conditions}\label{sect:cwc}

The non-comprehensive list below displays classical or recent winning conditions from computer science. It especially shows that new winning conditions obtained by conjunction of older winning conditions have been recently studied, e.g. in \cite{CD10} and \cite{BMRLL2016}.

{\bf Parity}  $\Col := \{0,1,\dots n\}$ for some $n \in \N$. A sequence is winning iff the least number occurring infinitely many times in the sequence is even.

{\bf Muller} $\Col := \{0,1,\dots n\}$ for some $n \in \N$. Let $M \subseteq \mathcal{P}(\Col)$ be a set of subsets of $\Col$. A sequence is winning iff the numbers occurring infinitely many times in the sequence constitute a set in $M$.

%{\bf Lim sup (Total payoff?)} $\Col = \R$, and a sequence is winning iff its limit superior is non-negative. (Variants exist with limit inferior or positivity)

{\bf Mean-payoff} $\Col = \R$, and a sequence is winning iff the limit superior of the partial sums is non-negative: $(u_n)_{n \in \N} \in \R^\N$ is winning iff $\limsup_{n \to _\infty}\frac{1}{n}\sum_{i=0}^nu_n \geq 0$. (Variants exist with limit inferior or positivity instead of non-negativity.)

{\bf  Energy} $\Col = \R$, and a sequence is winning iff its partial sums are non-negative: $(u_n)_{n \in \N} \in \R^\N$ is winning iff $\forall n \in \N,\,\sum_{i=0}^nu_n \geq 0$.

{\bf Boundedness}~\cite{LPR17} $\Col = \R^d$, and a sequence is winning iff its partial sums are uniformly bounded: $(u_n)_{n \in \N} \in (\R^d)^\N$ is winning iff $\exists b\forall n \in \N,\,\|\sum_{i=0}^nu_n\| \leq b$.

{\bf Discounted sum} $\Col$ is a bounded subset of $\R$. Let $0 < \alpha < 1$ and $t \in \R$. A sequence $(u_n)_{n \in \N} \in \R^\N$ is winning iff $\sum_{n=0}^{+\infty}\alpha^nu_n \geq t$.

{\bf Energy-parity}~\cite{CD10} $\Col := \R \times \{0,1,\dots n\}$ for some $n \in \N$. The winning condition is the conjunction of the energy (first component) and the parity (second component) conditions.

{\bf Average energy}~\cite{BMRLL2016} $\Col = \R$. The objective is to maintain a non-negative energy while keeping the average level of energy below a threshold $t \in \R$: a sequence $(u_n)_{n \in \N} \in \R^\N$ is winning iff $(\forall n \in \N,\,\sum_{i=0}^nu_n \geq 0) \,\wedge\, \limsup_{n \to +\infty}\frac{1}{n}\sum_{i=0}^n \sum_{j=0}^i u_j \leq t$.

\begin{observation}\label{obs:concrete-close-complete}
\begin{enumerate}
\item The parity, mean-payoff, energy, boundedness, energy-parity, and average energy conditions are all closed under interleaving. (It uses Lemma~\ref{lem:bool-combo}.\ref{lem:bool-combo2} to deals with energy-parity and average energy.)
\item The Muller and discounted sum conditions are not closed under interleaving.
\item The boundedness condition is factor-prefix complete; the others are not.
\item The energy, energy-parity, and discounted sum conditions are not closed under prefix removal; the others are.
\end{enumerate}
\end{observation}

\begin{corollary}
The turn-based safety-mean-payoff-parity games are half-positionally determined. (By Corollary~\ref{cor:half-pos-one} and Section~\ref{sect:apcc}.)
\end{corollary}

It may be disappointing that the Muller condition is not even closed under interleaving, but Proposition~\ref{prop:interleaving-win-one-eq} below extends Theorem~\ref{thm:interleaving-prefix removal-eq} to the concurrent Muller games. Using results from \cite{AH00} is likely to yield a better algorithmic complexity, though, but the point here is mainly that Theorem~\ref{thm:interleaving-prefix removal-eq} can be extended.
\begin{proposition}\label{prop:interleaving-win-one-eq}[Similar to \cite{AH00}]
Consider the finite games $\langle  \AO, \AT, Q, q_0, \delta, \Col, \tr, W\rangle$ where $W$ is a Muller condition. Deciding the existence of a \pO\/ winning (finite-memory) strategy can be done in big $O$ of 
$$
(|\AO||\AT||\Col||\Col|!)^2\cdot (|Q||\Col|^2|\Col|!)^{|Q||\Col||\Col|!\left(2^{|Q||\Col|^2|\Col|!}\right)}
$$
computation steps.
\end{proposition}

\subsection{Bounded residual load}\label{sect:brl}

Unlike Theorems~\ref{thm:interleaving-prefix removal-eq} and \ref {thm:interleaving-win-all-eq}, Theorems~\ref{thm:no-win-all-random-win} and \ref{thm:stateless-semi-rand-det} are not likely to be extended to include $\omega$-regular languages. Before defining a variant of the $\omega$-regular languages that satisfies the closeness and completeness properties from this article, let us consider notions of fairness that can be defined via a predicate $S$ on $\N \times \N \times \Col^\omega$. Intuitivily $S(n,d,\colHinf)$ is supposed to mean that the sequence $\colHinf$ has satisfied, with delay at most $d$, a request that was formulated in $\colHinf$ at time $n$.

There are several reasonable ways to express the good behavior of an infinite sequence using the $S(n,d,\colHinf)$. The classical definition of \emph{fairness} requires that all problems be eventually solved (see $F$ below), or cofinitely many problems (see $FCI$ below), for a usual weakening that ensures prefix independence of the condition. Arguing that this kind of fairness gives no guarantee about response time, \cite{AH00} strengthened fairness into \emph{finitary fairness}, which requires the existence of a uniform bound on the waiting time (see $FF$ below).

Yet another variant, \emph{bounded residual load} ($BRL$), is introduced below. It says that $\colHinf \in \Col^\omega$ satisfies $S$ wrt bounded residual load, if the number of problems that have currently not yet been solved is uniformly bounded overtime.
\begin{enumerate}
\item $F(\colHinf) := \forall n\in \N, \exists d \in \N,\,S(n,d,\colHinf)$
\item $FCI(\colHinf) := |\{n \in \N\,\mid\,\forall d \in \N,\,\neg S(n,d,\colHinf) \}| < \infty$
\item $FF(\colHinf) := \exists d \in \N, \forall n\in \N,\,S(n,d,\colHinf) $
\item $BRL(\colHinf) := \exists b \in \N, \forall n \in \N,\, \,b \geq |\{k \in \N\,\mid\, k \leq n\,\wedge\,\neg S(k,n-k,\colHinf)\}|$
\end{enumerate}

\begin{observation}\label{obs:brl-fairness}
\begin{enumerate}
\item $FF(\gamma)\,\Rightarrow\,F(\gamma)\quad \wedge \quad F(\gamma)\,\Rightarrow\,FCI(\gamma)$

\item $FF(\gamma)\,\Rightarrow\,BRL(\colHinf)\quad \wedge \quad BRL(\colHinf)\,\Rightarrow\,FCI(\gamma)$

\item $F$ and $BRL$ are incomparable in general.
\end{enumerate}
\end{observation}

The finitary fairness and the like may be too strict for some applications: gladly accepting to wait $b$ time units,  but categorically refusing to wait $b+1$ time units sounds unusual indeed. Instead, the system (which is responsible for solving the problems) could pay a penalty for each problem spending each time unit unsolved. Thanks to the bounded residual load, one has then the guarantee that the amount of money to be paid per time unit is bounded. 

It is possible to combine the two ideas, though:  by setting an acceptable response time and an acceptable uniform bound on the number of missed deadlines. This however, turns out to be equivalent to the simple $BRL$, which argues for the robustness of the concept.
\begin{observation}\label{obs:brl-richer-eq}
Let $BRLD(\colHinf) := \exists b,d \in \N, \forall n \in \N,\, \,b \geq |\{k \in \N\,\mid\, k \leq n -d \,\wedge\,\neg S(k,n-k,\colHinf)\}|$, then $BRLD(\colHinf)\,\Leftrightarrow\,BRL(\colHinf)$.\end{observation}

A second justification for the $BRL$ is that it has nice properties that the other notions of fairness lack when $S(n,d,\colHinf)$ is defined to minic $\omega$-regular languages, as shown below. Consider a non-empty set $\Col$ of colors and a function $\mathcal{C}: C \to \mathcal{P}(\Col^*)$. A sequence $\colHinf \in \Col^\omega$ is said to satisfy $\mathcal{C}$ from position $n$ after delay $d$, denoted $S_\mathcal{C}(n,d,\colHinf)$, if the following holds.
$$\exists u \in\mathcal{C}(\colHinf_n), \exists (k_1, \dots , k_{|u|}) \in \N^{|u|},\,n < k_1< \dots < k_{|u|} \leq n+d \,\wedge\,\forall i \leq |u|,\,u_i = \colHinf_{k_i}$$
Intuitively, each color is a problem or a request, and the problem may be solved in several ways, each way consisting in enumerating suitable colors quickly. (This might very well correspond to the positive fragment of some bounded-time temporal logic.) To simulate the parity condition, one can set $\Col := \N$ and $\mathcal{C}(2n) := \{\{k\} \mid  k \in \N \}$ and $\mathcal{C}(2n+1) := \{\{2k\} \mid  k \in \N\wedge k \leq n\}$ for all $n \in \N$. The corresponding $BRL_\mathcal{C}$ is the parity condition with bounded residual load.

Lemma~\ref{lem:brl-close-complete} below says that however $\mathcal{C}$ may be instantiated, all Theorems~\ref{thm:interleaving-prefix removal-eq}, \ref {thm:interleaving-win-all-eq}, \ref{thm:no-win-all-random-win}, and \ref{thm:stateless-semi-rand-det} can be applied with the $BRL_\mathcal{C}$ winning condition.

\begin{lemma}\label{lem:brl-close-complete}
For every non-empty set $\Col$ of colors and every function $\mathcal{C}: C \to \mathcal{P}(\N^\Col)$, the winning condition $BRL_\mathcal{C}$ is closed under prefix removal  and interleaving, and factor-prefix complete.
\end{lemma}

Even when $\mathcal{C}$ simulates the parity condition as above, none of the corresponding $F_\mathcal{C}$, $FCI_\mathcal{C}$, or $FF_\mathcal{C}$ is both closed under interleaving and factor-set complete. $FF_\mathcal{C}$ is not closed under interleaving: $FF_\mathcal{C}((01)^\omega)$ and $FF_\mathcal{C}((23)^\omega)$, but $\neg FF_\mathcal{C}(\colHinf)$, where $\colHinf := (23)01(23)^201 \dots 01(23)^n01 \dots$ can be obtained by interleaving $(01)^\omega$ and $(23)^\omega$. $FCI_\mathcal{C}$ is not factor set-complete: $FCI_\mathcal{C}(\colHinf)$, where $\colHinf := 101^201^3 \dots 01^n0 \dots$, but $\neg FCI_\mathcal{C}(1^\omega)$ altough factors of $1$'s occur with arbitrary length in $\colHinf$. $F_\mathcal{C}$ is neither: first, $F_\mathcal{C}((10)^\omega)$ and $F_\mathcal{C}(2^\omega)$, but $\neg F_\mathcal{C}(1 \cdot 2^\omega)$, altough $1 \cdot 2^\omega$ can be obtained by interleaving $(10)^\omega$ and $2^\omega$; second, as above for $FCI_\mathcal{C}$. Note that the window-parity condition~\cite{CDRR15},\cite{BHR16} is not closed under interleaving either, as again exemplified by $(01)^\omega$ and $(23)^\omega$.

%%
%% Bibliography
%%

%% Either use bibtex (recommended), 

\bibliography{article}

\begin{thebibliography}{10}

\bibitem{AFKS16}
Rajeev Alur, Marco Faella, Sampath Kannan, and Nimit Singhania.
\newblock {Hedging Bets in Markov Decision Processes}.
\newblock In Jean-Marc Talbot and Laurent Regnier, editors, {\em 25th EACSL
  Annual Conference on Computer Science Logic (CSL 2016)}, volume~62 of {\em
  Leibniz International Proceedings in Informatics (LIPIcs)}, pages
  29:1--29:20, Dagstuhl, Germany, 2016. Schloss Dagstuhl--Leibniz-Zentrum fuer
  Informatik.
\newblock URL: \url{http://drops.dagstuhl.de/opus/volltexte/2016/6569}, \href
  {http://dx.doi.org/10.4230/LIPIcs.CSL.2016.29}
  {\path{doi:10.4230/LIPIcs.CSL.2016.29}}.

\bibitem{AH98}
Rajeev Alur and Thomas~A. Henzinger.
\newblock Finitary fairness.
\newblock {\em ACM Trans. Program. Lang. Syst.}, 20(6):1171--1194, November
  1998.
\newblock URL: \url{http://doi.acm.org/10.1145/295656.295659}, \href
  {http://dx.doi.org/10.1145/295656.295659} {\path{doi:10.1145/295656.295659}}.

\bibitem{BBMU12}
Patricia Bouyer, Romain Brenguier, Nicolas Markey, and Michael Ummels.
\newblock Concurrent games with ordered objectives.
\newblock In Lars Birkedal, editor, {\em Foundations of Software Science and
  Computational Structures}, pages 301--315, Berlin, Heidelberg, 2012. Springer
  Berlin Heidelberg.

\bibitem{BMRLL2016}
Patricia Bouyer, Nicolas Markey, Mickael Randour, Kim~G. Larsen, and Simon
  Laursen.
\newblock Average-energy games.
\newblock {\em Acta Informatica}, Jul 2016.
\newblock URL: \url{https://doi.org/10.1007/s00236-016-0274-1}, \href
  {http://dx.doi.org/10.1007/s00236-016-0274-1}
  {\path{doi:10.1007/s00236-016-0274-1}}.

\bibitem{BHR16}
V{\'{e}}ronique Bruy{\`{e}}re, Quentin Hautem, and Mickael Randour.
\newblock Window parity games: an alternative approach toward parity games with
  time bounds.
\newblock In {\em Proceedings of the Seventh International Symposium on Games,
  Automata, Logics and Formal Verification, GandALF 2016, Catania, Italy, 14-16
  September 2016.}, pages 135--148, 2016.
\newblock URL: \url{https://doi.org/10.4204/EPTCS.226.10}, \href
  {http://dx.doi.org/10.4204/EPTCS.226.10} {\path{doi:10.4204/EPTCS.226.10}}.

\bibitem{Chatterjee07}
Krishnendu Chatterjee.
\newblock Concurrent games with tail objectives.
\newblock {\em Theoretical Computer Science}, 388(1):181 -- 198, 2007.
\newblock URL:
  \url{http://www.sciencedirect.com/science/article/pii/S0304397507005804},
  \href {http://dx.doi.org/https://doi.org/10.1016/j.tcs.2007.07.047}
  {\path{doi:https://doi.org/10.1016/j.tcs.2007.07.047}}.

\bibitem{CD10}
Krishnendu Chatterjee and Laurent Doyen.
\newblock Energy parity games.
\newblock In Samson Abramsky, Cyril Gavoille, Claude Kirchner, Friedhelm Meyer
  auf~der Heide, and Paul~G. Spirakis, editors, {\em Automata, Languages and
  Programming}, pages 599--610, Berlin, Heidelberg, 2010. Springer Berlin
  Heidelberg.

\bibitem{CDRR15}
Krishnendu Chatterjee, Laurent Doyen, Mickael Randour, and
  Jean{-}Fran{\c{c}}ois Raskin.
\newblock Looking at mean-payoff and total-payoff through windows.
\newblock {\em Inf. Comput.}, 242:25--52, 2015.
\newblock URL: \url{https://doi.org/10.1016/j.ic.2015.03.010}, \href
  {http://dx.doi.org/10.1016/j.ic.2015.03.010}
  {\path{doi:10.1016/j.ic.2015.03.010}}.

\bibitem{CHH09}
Krishnendu Chatterjee, Thomas~A. Henzinger, and Florian Horn.
\newblock Finitary winning in {$\Omega$}-regular games.
\newblock {\em ACM Trans. Comput. Logic}, 11(1):1:1--1:27, November 2009.
\newblock URL: \url{http://doi.acm.org/10.1145/1614431.1614432}, \href
  {http://dx.doi.org/10.1145/1614431.1614432}
  {\path{doi:10.1145/1614431.1614432}}.

\bibitem{CN06}
Thomas Colcombet and Damian Niwiński.
\newblock On the positional determinacy of edge-labeled games.
\newblock {\em Theoretical Computer Science}, 352(1):190 -- 196, 2006.
\newblock URL:
  \url{http://www.sciencedirect.com/science/article/pii/S0304397505008303},
  \href {http://dx.doi.org/https://doi.org/10.1016/j.tcs.2005.10.046}
  {\path{doi:https://doi.org/10.1016/j.tcs.2005.10.046}}.

\bibitem{AH00}
Luca de~Alfaro and Thomas~A. Henzinger.
\newblock Concurrent omega-regular games.
\newblock In {\em 15th Annual {IEEE} Symposium on Logic in Computer Science,
  Santa Barbara, California, USA, June 26-29, 2000}, pages 141--154, 2000.
\newblock URL: \url{https://doi.org/10.1109/LICS.2000.855763}, \href
  {http://dx.doi.org/10.1109/LICS.2000.855763}
  {\path{doi:10.1109/LICS.2000.855763}}.

\bibitem{DJP03}
Nachum Dershowitz, D.~N. Jayasimha, and Seungjoon Park.
\newblock {\em Bounded Fairness}, pages 304--317.
\newblock Springer Berlin Heidelberg, Berlin, Heidelberg, 2003.
\newblock URL: \url{https://doi.org/10.1007/978-3-540-39910-0_14}, \href
  {http://dx.doi.org/10.1007/978-3-540-39910-0_14}
  {\path{doi:10.1007/978-3-540-39910-0_14}}.

\bibitem{GH10}
Hugo Gimbert and Florian Horn.
\newblock Solving simple stochastic tail games.
\newblock In {\em Proceedings of the Twenty-first Annual ACM-SIAM Symposium on
  Discrete Algorithms}, SODA '10, pages 847--862, Philadelphia, PA, USA, 2010.
  Society for Industrial and Applied Mathematics.
\newblock URL: \url{http://dl.acm.org/citation.cfm?id=1873601.1873670}.

\bibitem{GZ05}
Hugo Gimbert and Wies{\l}aw Zielonka.
\newblock Games where you can play optimally without any memory.
\newblock In {\em CONCUR 2005 - Concurrency Theory}, volume 3653 of {\em
  Lecture Notes in Computer Science}, pages 428--442. Springer Berlin
  Heidelberg, 2005.

\bibitem{GMPRW17}
Julian Gutierrez, Aniello Murano, Giuseppe Perelli, Sasha Rubin, and Michael
  Wooldridge.
\newblock Nash equilibria in concurrent games with lexicographic preferences.
\newblock In {\em Proceedings of the Twenty-Sixth International Joint
  Conference on Artificial Intelligence, {IJCAI} 2017, Melbourne, Australia,
  August 19-25, 2017}, pages 1067--1073, 2017.
\newblock URL: \url{https://doi.org/10.24963/ijcai.2017/148}, \href
  {http://dx.doi.org/10.24963/ijcai.2017/148}
  {\path{doi:10.24963/ijcai.2017/148}}.

\bibitem{HP79}
John Hopcroft and Jean-Jacques Pansiot.
\newblock On the reachability problem for 5-dimensional vector addition
  systems.
\newblock {\em Theoretical Computer Science}, 8(2):135 -- 159, 1979.
\newblock URL:
  \url{http://www.sciencedirect.com/science/article/pii/0304397579900410},
  \href {http://dx.doi.org/https://doi.org/10.1016/0304-3975(79)90041-0}
  {\path{doi:https://doi.org/10.1016/0304-3975(79)90041-0}}.

\bibitem{KM69}
Richard~M. Karp and Raymond~E. Miller.
\newblock Parallel program schemata.
\newblock {\em Journal of Computer and System Sciences}, 3(2):147 -- 195, 1969.
\newblock URL:
  \url{http://www.sciencedirect.com/science/article/pii/S0022000069800115},
  \href {http://dx.doi.org/https://doi.org/10.1016/S0022-0000(69)80011-5}
  {\path{doi:https://doi.org/10.1016/S0022-0000(69)80011-5}}.

\bibitem{Kopczynski06}
Eryk Kopczy{\'{n}}ski.
\newblock Half-positional determinacy of infinite games.
\newblock In Michele Bugliesi, Bart Preneel, Vladimiro Sassone, and Ingo
  Wegener, editors, {\em Automata, Languages and Programming}, pages 336--347,
  Berlin, Heidelberg, 2006. Springer Berlin Heidelberg.

\bibitem{SLR13}
St\'ephane Le~Roux.
\newblock Infinite sequential {N}ash equilibrium.
\newblock {\em Logical Methods in Computer Science}, 9, 2013.
\newblock Special Issue for the Conference "Computability and Complexity in
  Analysis" (CCA 2011).

\bibitem{SLR14}
St\'ephane Le~Roux.
\newblock From winning strategy to {N}ash equilibrium.
\newblock {\em Mathematical Logic Quarterly}, 2014.
\newblock (to appear).

\bibitem{SLR15}
St{\'{e}}phane Le~Roux.
\newblock Infinite subgame perfect equilibrium in the hausdorff difference
  hierarchy.
\newblock In {\em Topics in Theoretical Computer Science - The First {IFIP}
  {WG} 1.8 International Conference, {TTCS} 2015, Tehran, Iran, August 26-28,
  2015, Revised Selected Papers}, pages 147--163, 2015.
\newblock URL: \url{https://doi.org/10.1007/978-3-319-28678-5_11}, \href
  {http://dx.doi.org/10.1007/978-3-319-28678-5_11}
  {\path{doi:10.1007/978-3-319-28678-5_11}}.

\bibitem{LP14}
St{\'e}phane Le~Roux and Arno Pauly.
\newblock Infinite sequential games with real-valued payoffs.
\newblock In {\em Proceedings of LiCS}, 2014.

\bibitem{LP17}
St{\'{e}}phane Le~Roux and Arno Pauly.
\newblock Extending finite memory determinacy to multiplayer games.
\newblock In {\em Proceedings of the 4th International Workshop on Strategic
  Reasoning, {SR} 2016, New York City, USA, 10th July 2016.}, pages 27--40,
  2016.
\newblock URL: \url{https://doi.org/10.4204/EPTCS.218.3}, \href
  {http://dx.doi.org/10.4204/EPTCS.218.3} {\path{doi:10.4204/EPTCS.218.3}}.

\bibitem{LPR17}
St{\'{e}}phane Le~Roux, Arno Pauly, and Jean{-}Fran{\c{c}}ois Raskin.
\newblock Minkowski games.
\newblock In {\em 34th Symposium on Theoretical Aspects of Computer Science,
  {STACS} 2017, March 8-11, 2017, Hannover, Germany}, pages 50:1--50:13, 2017.
\newblock URL: \url{https://doi.org/10.4230/LIPIcs.STACS.2017.50}, \href
  {http://dx.doi.org/10.4230/LIPIcs.STACS.2017.50}
  {\path{doi:10.4230/LIPIcs.STACS.2017.50}}.

\bibitem{Martin75}
Donald~A. Martin.
\newblock Borel determinacy.
\newblock {\em Annals of Mathematics}, 102(2):363--371, 1975.
\newblock URL: \url{http://www.jstor.org/stable/1971035}.

\bibitem{Martin98}
Donald~A. Martin.
\newblock The determinacy of blackwell games.
\newblock {\em Journal of Symbolic Logic}, 63(4):1565–1581, 1998.
\newblock \href {http://dx.doi.org/10.2307/2586667}
  {\path{doi:10.2307/2586667}}.

\bibitem{LP18}
Stéphane~Le Roux and Arno Pauly.
\newblock Extending finite-memory determinacy to multi-player games.
\newblock {\em Information and Computation}, pages~--, 2018.
\newblock URL:
  \url{https://www.sciencedirect.com/science/article/pii/S0890540118300270},
  \href {http://dx.doi.org/https://doi.org/10.1016/j.ic.2018.02.024}
  {\path{doi:https://doi.org/10.1016/j.ic.2018.02.024}}.

\bibitem{Thomas97}
Wolfgang Thomas.
\newblock Languages, automata, and logic.
\newblock In Salomaa~A. Rozenberg~G., editor, {\em Handbook of Formal
  Languages}. Springer, Berlin, Heidelberg, 1997.

\end{thebibliography}

%% .. or use the thebibliography environment explicitely

\newpage

\appendix

\section{Existence of \pO\/ winning strategies}

\subsection{More on interleaving}

\emph{Interleaving} two finite word consists in enumerating sequentially the two words to produce a new word. For example, interleaving $024$ and $13$ can produce the words $01234$ and $10324$, but neither $31024$ nor $01432$. Formally, interleaving finite words over some alphabet $\Col$ is defined by induction: for all $\colH,\colH',\colH'' \in \Col^*$ and $\col \in \Col$, set $\e \in I(\colH,\colH')$ and $\colH'' \in I(\colH,\colH')\,\Rightarrow \colH'' \col \in I(\colH \col,\colH') \cap I(\colH,\colH' \col)$.

\begin{observation}
Interleaving finite words is associative and commutative.
\end{observation}

Let us now give a possible formalization of the interleaving of infinite words. Let $\colHinf^1, \dots,\colHinf^k \in \Col^\omega$. Then $\colHinf \in I(\colHinf^1,\dots,\colHinf^k)$ iff there exist $l_1, \dots,l_k:\N \to \N$ such that
\begin{itemize}
\item $\forall n\in N,\, \sum_{i=1}^k l_i(n)= n$,
\item $\forall n\in N,\forall i \in \{1,\dots,k\},\,l_i(n) \leq l_i(n+1)$,
\item $\forall n\in N,\,\colHinf_n = \colHinf^{i(n)}_{l_{i(n)}(n)} \mbox{ where }i(n) \mbox{ is the unique }i\mbox{ such that }l_{i}(n) < l_{i}(n+1)$.
\end{itemize}

\begin{lemma}\label{lem:interleaving-eq}
Let $\colHinf^1, \dots,\colHinf^k \in \Col^\omega$. Then $\colHinf \in I(\colHinf^1,\dots,\colHinf^k)$ iff for all $n \in \N$ there exists $(\colH^{i})_{i \leq k}$ such that $\colH^{i} \sqsubset \colHinf^{i}$ for all $i \leq k$ and $\colHinf_{< n}$ can be obtained by interleaving the $\colH^i$.

\begin{proof}
Let us first assume that $\colHinf \in I(\colHinf^1,\dots,\colHinf^k)$, and let us prove by induction on $n \in \N$ that $\colHinf_{< n} \in I(\colHinf^1_{< l_1(n)},\dots, \colHinf^k_{< l_k(n)})$. For the base case, $\colHinf_{< 0} = \e \in I(\e,\dots, \e) = I(\colHinf^1_{< l_1(0)},\dots, \colHinf^k_{< l_k(0)})$ since $l_i(0) = 0$ for all $i \leq k$. For the inductive case $\colHinf_{< n+1}  = \colHinf_{< n}\colHinf_{n} = \colHinf_{< n}\colHinf^{i(n)}_{l_{i(n)}(n)}$ by definition of the interleaving of infinite words. By I.H $\colHinf_{< n}\in I(\colHinf^1_{< l_1(n)},\dots, \colHinf^k_{< l_k(n)})$. On the one hand, for all $j \neq i(n)$, from $l_j(n+1) = l_j(n)$ follows $\colHinf^j_{l_j(n+1)} = \colHinf^j_{l_j(n)}$; on the other hand, $\colHinf^{i(n)}_{< l_{i(n)}(n+1)} = \colHinf^{i(n)}_{< l_{i(n)}(n)}\colHinf^{i(n)}_{l_{i(n)}(n)}$ since $l_{i(n)}(n) +1 = l_{i(n)}(n+1)$. Therefore $\colHinf_{< n+1} \in I(\colHinf^1_{< l_1(n+1)},\dots, \colHinf^k_{< l_k(n+1)})$ by definition of interleaving of finite words.

Conversely, let us assume that for all $n \in \N$ there exists $(\colH^{i})_{i \leq k}$ such that $\colH^{i} \sqsubset \colHinf^{i}$ for all $i \leq k$ and $\colHinf_{< n}$ can be obtained by interleaving the $\colH^i$.
\end{proof}
\end{lemma}

\subsection{More on Lemma~\ref{lem:fpi}}

Below is a short story that might help provide useful insight to some readers. Once upon a time, there was a capricious king who loved pastry. There were many bakeries in his kingdom, and each of them could bake a wide range of delicious cakes. Each shop would bake only one type of cake per day, though, and the only way to know which was to visit the shop. One morning, the king summoned his minister to bring him his favorite cake for dinner (among the cakes of the day). Unfortunately, the shops were far apart and one could only visit one of them within a day, and the king's favorite depended on the cakes of the day in an irrational way. The minister considered buying a cake from some shop and lying about the cakes of the day in the other shops. But the king knew the range of each shop, what if there were no plausible lie? Desperate, the minister sought help from a mathematician:  she enquired about the king's preferences and the range of each shop, bought a cake from one shop, lied about the cakes of the day, and the king ate happily. Lemma~\ref{lem:fpi} shows that the mathematician was bound to succeed: given the ranges of the shops and the king's preferences, there always exists a safe shop. T%o make Lemma~\ref{lem:fpi} as readable a possible, the natural numbers are also used as von Neumann ordinals, i.e. $n$ may represent the set $\{0,1,\dots, n-1\}$. Below, $k+1$ is the number of backeries, $n_i$ is the range of the $i$-th shop, $f$ tells the king's favorite (shop of the day),  $\exists i$ means that there exists a safe shop, $j \in n_i$ is the cake of the day in the safe shop, $u$ is the fake list, $u_i = j$ means that the fake list is consistent with the cake of the day in the safe shop, and $f(u) = i$ means that this cake is indeed the king's favorite on the fake list.

Corollary~\ref{cor:fpi} below is derived from Lemma~\ref{lem:fpi} by partial Skolemization, i.e. by  pulling the $\exists i$ before the $\forall f$, and the $\exists u$ before the $\forall j$, thus automatically yielding the $\exists F$ and the $\exists G_f$, respectively. Whereas Lemma~\ref{lem:fpi} could be invoked to characterize the existence of winning strategies, Corollary~\ref{cor:fpi} will be invoked to characterize the existence of winning strategies with (finite) memory, which will be constructed via the functions $F$ and $G_f$. Note that in the statement of Corollary~\ref{cor:fpi} uses natural numbers as  von Neumann ordinals.
\begin{corollary}\label{cor:fpi}
$\forall k \in \N, \forall n_0,\dots, n_k \in (\N \setminus \{0\})^{k+1}, \exists F: (n_0 \times \dots \times n_k \to k+1) \to k+1, \forall f : n_0 \times \dots \times n_k \to k+1, \exists G_f : n_{F(f)} \to n_0 \times \dots \times n_k, \forall j \in n_{F(f)}, G_f(j)_{F(f)} = j \wedge f(u) = F(f)$.
\end{corollary}

\subsection{Using Lemma~\ref{lem:fpi}}

Lemma~\ref{lem:fpi} is then used in Lemma~\ref{lem:win-delayed-response-implies-win} which factors out most of the proof burden of Theorems~\ref{thm:interleaving-prefix removal-eq} and \ref {thm:interleaving-win-all-eq}. Lemma~\ref{lem:win-delayed-response-implies-win} involves games that are concurrent at fewer states than in the original game, and then the proofs of Theorems~\ref{thm:interleaving-prefix removal-eq} and \ref {thm:interleaving-win-all-eq} proceed by induction on the degree of concurrency. To define these simpler games, consider $g = \langle  \AO, \AT, Q, q_0, \delta, \Col, \tr, W\rangle$.

{\bf States involving a player}
A state $q \in Q$ is said to involve \pT\/ if $1 < |(\tr,\delta)(q,\aO,\AT)|$ for some $\aO \in \AO$. Indeed, if $1 = |(\tr,\delta)(q,\aO,\AT)|$ for all $\aO \in \AO$, at state $q$ the action chosen by \pT\/ is irrelevant to the produced color and next state.

{\bf Delayed $q$-responses}
Let $q \in Q$. Informally, a delayed $q$-response is a partial delayed response only defined at state $q$. Formally, let $E_1,\dots,E_{k}$ be the elements of $\{(\tr,\delta)(q,\aO,\AT)\mid\aO \in \AO\}$. The elements of $E_1 \times \dots \times E_k$ are called the delayed $q$-responses in $g$.

{\bf $q$-derived games}
Let $q \in Q$ and let $\overline{e}$ be a delayed $q$-response $\overline{e}$ in $g$. Informally, $g(q,\overline{e})$ is a game derived from $g$ by modifying the local interaction at state $q$, such that the resulting game, called a $q$-derived game, is a mix between $g$ and some derived game of $g$. Formally, $g(q,\overline{e}) := \langle  \AO, \AT, Q, q_0, \delta', \Col, \tr', W\rangle$, where  for all $(\aO,\aT, q') \in \AO \times \AT \times (Q \setminus \{q\})$ it is set $(\tr',\delta')(q',\aO,\aT) := (\tr,\delta)(q',\aO,\aT)$ and $(\tr',\delta')(q,\aO,\aT) := \overline{e}_{(\tr,\delta)(q,\aO,\AT)}$.

Lemma~\ref{lem:2p-t-1p} may sound a bit technical, partly because it is meant to be used in two slightly different contexts. It is used once to prove the \ref{thm:interleaving-prefix removal-eq2} $\Rightarrow$\ref{thm:interleaving-prefix removal-eq1} implication of Theorem~\ref{thm:interleaving-prefix removal-eq}, where $F(q,\aO)$ (used in Lemma~\ref{lem:2p-t-1p}) is the full set $\AT$; and it is also used in the proof of Lemma~\ref{lem:win-implies-win-delayed-response}, where the $F(q,\aO)$ are singletons. Then Lemma~\ref{lem:win-implies-win-delayed-response} is used to prove the \ref{thm:interleaving-prefix removal-eq1} $\Rightarrow$\ref{thm:interleaving-prefix removal-eq2} implication of Theorem~\ref{thm:interleaving-prefix removal-eq} and the \ref{thm:interleaving-win-all-eq1} $\Rightarrow$\ref{thm:interleaving-win-all-eq2} implication of Theorem~\ref{thm:interleaving-win-all-eq}.

\begin{lemma}\label{lem:2p-t-1p}
Consider a game $g_{q_0} = \langle  \AO, \AT, Q, q_0, \delta, \Col, \tr, W\rangle$. Let $t$ be a delayed response  in $g_{q_0}$, and let $F: Q \times \AO \to \mathcal{P}(\AT)$ be such that $(\tr,\delta)(q,\aO,F(q,\aO)) = \{t_{q,(\tr,\delta)(q,\aO,\AT)}\}$ for all $(q,\aO) \in Q \times \AO$. Let $\sO$ be a \pO\/ strategy in $g_{q_0}$, let $\aHTinf \in \AT^\omega$ be such that $\aHTinf_{n} \in F(\Delta \circ \hOS(\sO,\aHTinf_{<n}),\sO(\aHTinf_{<n}))$ for all $n \in \N$, and let $\sO_t$ be the \pO\/ strategy in $g_{q_0}(t)$ such that $\sO_t(n) = \sO(\aHTinf_{<n})$ for all $n \in \N$. Then $\tr \circ \hOS(\sO,\aHTinf) = \tr_t \circ \hOS_t(\sO_t,\omega)$.
\begin{proof}
For all  $(q,\aO) \in Q \times \AO$, $(\tr,\delta)(q,\aO,F(q,\aO)) = \{t_{q,(\tr,\delta)(q,\aO,\AT)}\}$ by assumption on $F$, and $(\tr_t,\delta_t)(q,\aO) = t_{q,(\tr,\delta)(q,\aO,\AT)}$ by definition of $g_{q_0}(t)$, so
\begin{align}
\forall (q,\aO) \in Q \times \AO,\,\{(\tr_t,\delta_t)(q,\aO)\} = (\tr,\delta)(q,\aO,F(q,\aO)) \label{eq:2p-t-1p1}
\end{align}
Intuitively, Equation~(\ref{eq:2p-t-1p1}) suggests that \pT\/ can simulate $g_{q_0}(t)$ in $g_{q_0}$ by always guessing which action \pO\/ is going to choose and then by choosing his own action accordingly, i.e. via $F$. To give a formal content to this intuition, let us first prove that the sequences of the visited states are the same in the two games, if \pT\/ follows $\aHTinf$ in $g_{q_0}$. More specifically, let us prove the following by induction on $n$.
\begin{align}
\forall n \in \N,\, \Delta_t \circ \hOS_t(\sO_t,n) = \Delta \circ \hOS(\sO,\aHTinf_{< n})\label{eq:2p-t-1p2}
\end{align}
For the base case, $\Delta_t \circ \hOS_t(\sO_t,0) = \Delta_t (\e) = q_0 =\Delta(\e) = \Delta \circ \hOS(\sO,\aHTinf_{< 0})$. For the inductive case, 
\begin{align*}
\Delta_t \circ  \hOS_t(\sO_t,n+1) & = \Delta_t(\hOS_t(\sO_t,n) \cdot \sO_t(n)) \mbox{ by definition of }\hOS_t,\\
	& = \delta_t(\Delta_t \circ  \hOS_t(\sO_t,n),\sO_t(n))  \mbox{ by definition of }\Delta_t,\\
	& = \delta_t(\Delta \circ \hOS(\sO,\aHTinf_{<n}),\sO_t(n)) \mbox{ by I.H.,}\\
	& = \delta_t(\Delta \circ \hOS(\sO,\aHTinf_{<n}), \sO(\aHTinf_{<n})) \mbox{ by definition of }\sO_t,\\
	& = (\{x\} \mapsto x)\\
	&\quad  \left(\delta(\Delta \circ \hOS(\sO,\aHTinf_{<n}), \sO(\aHTinf_{<n}),F(\Delta \circ \hOS(\sO,\aHTinf_{<n}), \sO(\aHTinf_{<n})))\right) \mbox{ by Eq.~(\ref{eq:2p-t-1p1}),}\\
	& = \delta(\Delta \circ \hOS(\sO,\aHTinf_{<n}), \sO(\aHTinf_{<n}),\aHTinf_{n}) \mbox{ by definition of }\aHTinf,\\
	& = \Delta (\hOS(\sO,\aHTinf_{<n}) \cdot (\sO(\aHTinf_{<n}),\aHTinf_{n}))\mbox{ by definition of }\Delta,\\	
	& = \Delta \circ \hOS(\sO,\aHTinf_{<n+1})\mbox{ by definition of }\hOS\mbox{, thus completing the induction.}
\end{align*}

Next, let us prove that the sequences of produced colors are the same in the two games, if \pT\/ follows $\aHTinf$ in $g_{q_0}$. More specifically, let us also prove the following by induction on $n$.
\begin{align}
\forall n \in \N,\, \tr_t \circ \hOS_t(\sO_t,n) = \tr \circ \hOS(\sO,\aHTinf_{< n}) \label{eq:2p-t-1p3}
\end{align}
For the base case, $\tr_t \circ \hOS_t(\sO_t,0) = \tr_t (\e) = \e = \tr(\e) =  \tr \circ \hOS(\sO,\aHTinf_{< 0})$. For the inductive case, 
\begin{align*}
\tr_t \circ  \hOS_t(\sO_t,n+1) & = \tr_t(\hOS_t(\sO_t,n) \cdot \sO_t(n)) \mbox{ by definition of }\hOS_t,\\
	& = \tr_t \circ \hOS_t(\sO_t,n) \cdot \tr_t(\Delta_t \circ \hOS_t(\sO_t,n), \sO_t(n)) \mbox{ by definition of }\tr_t,\\
	& =  \tr \circ \hOS(\sO,\aHTinf_{<n}) \cdot \tr_t(\Delta_t \circ \hOS_t(\sO_t,n), \sO_t(n)) \mbox{ by I.H.,}\\
	& =  \tr \circ \hOS(\sO,\aHTinf_{<n}) \cdot \tr_t(\Delta \circ \hOS(\sO,\aHTinf_{<n}), \sO_t(n)) \mbox{ by Equation~(\ref{eq:2p-t-1p2}),}\\	
	& = \tr \circ \hOS(\sO,\aHTinf_{<n}) \cdot \tr_t(\Delta \circ \hOS(\sO,\aHTinf_{<n}), \sO(\aHTinf_{<n})) \mbox{ by definition of }\sO_t\\
	& = \tr \circ \hOS(\sO,\aHTinf_{<n}) \cdot (\{x\}\mapsto x)\\
	& \quad \left(\tr(\Delta \circ \hOS(\sO,\aHTinf_{<n}), \sO(\aHTinf_{<n}), F(\Delta \circ \hOS(\sO,\aHTinf_{<n}), \sO(\aHTinf_{<n})))\right) \mbox{ by Eq.~(\ref{eq:2p-t-1p1}),}\\
	& = \tr \circ \hOS(\sO,\aHTinf_{<n}) \cdot \tr(\Delta \circ \hOS(\sO,\aHTinf_{<n}), \sO(\aHTinf_{<n}),\aHTinf_{n}) \mbox{ by definition of }\aHTinf,\\
	& =  \tr ( \hOS(\sO,\aHTinf_{<n}) \cdot (\sO(\aHTinf_{<n}),\aHTinf_{n})) \mbox{ by definition of }\tr,\\
	& = \tr \circ \hOS(\sO,\aHTinf_{<n+1})\mbox{ by definition of }\hOS\mbox{, thus completing the induction.}
\end{align*}
Finally, Equation~(\ref{eq:2p-t-1p3}) may be lifted to infinite arguments: $\tr_t \circ \hOS_t(\sO_t,\omega) = \tr \circ \hOS(\sO,\aHTinf)$.
\end{proof}
\end{lemma}

\begin{lemma}\label{lem:win-implies-win-delayed-response}
Consider a game $g_{q_0} = \langle  \AO, \AT, Q, q_0, \delta, \Col, \tr, W\rangle$. If \pO\/ wins $g_{q_0}$ (with finite memory),  \pO\/ wins $g_{q_0}(t)$ (with finite memory) for all delayed responses $t$.
\begin{proof}
Let $q_0 \in Q$ and $t$ be a delayed response in $g_{q_0}$. By definition of the delayed responses, $t_{q,(\tr,\delta)(q,\aO,\AT)} \in (\tr,\delta)(q,\aO,\AT)$ for all $(q,\aO) \in Q \times \AO$, so by the axiom of choice there exists a function $f: Q \times \AO \to \AT$ such that $(\tr,\delta)(q,\aO,f(q,\aO)) = t_{q,(\tr,\delta)(q,\aO,\AT)}$ for all $(q,\aO) \in Q \times \AO$.

Let $\sO$ be a \pO\/ winning strategy in $g_{q_0}$, and let $\aHTinf \in \AT^\omega$ be defined by $\aHTinf_{n} := f(\Delta \circ \hOS(\sO,\aHTinf_{< n}),\sO(\aHTinf_{< n}))$. Let a \pO\/ strategy $\sO_t$ in $g_{q_0}(t)$ be defined by $\sO_t(n) := \sO(\aHTinf_{< n})$ for all $n \in \N$. 

Invoking Lemma~\ref{lem:2p-t-1p} with $F(q,\aO) := \{f(q,\aO)\}$ yields $\tr_t \circ \hOS_t(\sO_t,\omega) = \tr \circ \hOS(\sO,\aHTinf)$. Since $\sO$ is winning in $g_{q_0}$, $\tr \circ \hOS(\sO,\aHTinf) \in W$, so $\tr_t \circ \hOS_t(\sO_t,\omega) \in W$, and $\sO_t$ is winning in $g_{q_0}(t)$.
\end{proof}
\end{lemma}

\begin{lemma}\label{lem:win-delayed-response-implies-win}
Consider games $g_{q} = \langle  \AO, \AT, Q, q, \delta, \Col, \tr, W\rangle$ parametrized by $q \in Q$, where $Q$ and $\Col$ are finite and $W$ is interleaving-closed. Let $q_0,q_1 \in Q$, let $\overline{e}^1$ be a delayed $q_1$-response, and let $m \in \N$. Let us assume that there exists a \pO\/ winning strategy in $g_{q_0}(q_1,\overline{e}^1)$ using memory size $m$, and that for all delayed $q_1$-responses $\overline{e} \neq \overline{e}^1$ there exists a \pO\/ winning strategy in $g_{q_1}(q_1,\overline{e})$ using memory size $m$. Then there exists a  \pO\/ winning strategy in $g_{q_0}$ using memory size $dm + \log_2 d$, where $d$ is the number of delayed $q_1$-responses. 
\begin{proof}[Proof of Lemma~\ref{lem:win-delayed-response-implies-win}]
{\bf Sketch}
Consider the game $g'$ played almost like $g_{q_0}$, but the rounds played at the state $q_1$ are no longer concurrent, they are split into two subrounds: in the first subround \pT\/ chooses some delayed $q_1$-response $\overline{e}$; in the second subround \pO\/ chooses some projection $\overline{e}_i$ thereof, thus producing a color and inducing the next state. \pO\/ can win this game by pretending that she is playing some $q$-derived games (the ones from the statement of Lemma~\ref{lem:win-delayed-response-implies-win}) in an interleaved way: \pT\/ is choosing the interleaving order by switching (or not) games when at state $q_1$, and \pO\/ resumes the play of the relevant game where it was last interrupted. Since $W$ is interleaving-closed and the run in $g'$ is obtained by interleaving runs from $q$-derived games all won by \pO\/, the induced run is also in $W$.

For \pO\/, playing $g_{q_0}$ amounts to playing $g'$ with a significant handicap: she may no longer know what \pT\/ is choosing at state $q_1$, i.e. if the current state is $q_1$, she may not know which $q$-derived game she is supposed to play. When at $q_1$, the best she can hope for is to determine some $\aO \in \AO$ such that, regardless of which $\aT \in \AT$ \pT\/ may choose, $(\tr,\delta)(q,\aO,\aT)$ are the very color and next state that are expected in order to proceed with the play of some underlying $q$-derived game. The existence of a suitable $\aO \in \AO$ follows from Lemma~\ref{lem:fpi}/Corollary~\ref{cor:fpi}.

{\bf Details}
Let $E_1,\dots,E_k$ be the elements of the set $\{(\tr,\delta)(q_1,\aO,\AT)\,\mid\, \aO \in \AO\}$, so $\overline{e}^1 \in E_1 \times \dots \times E_k$. By assumption, let $\sO_{\overline{e}^1}$ be a \pO\/ winning strategy in $g_{q_0}(q_1,\overline{e}^1)$, and for all $\overline{e} \in (E_1 \times \dots \times E_k) \setminus \{\overline{e}^1\}$, let $\sO_{\overline{e}}$ be a \pO\/ winning strategy in $g_{q_1}(q_1,\overline{e})$. For all $i \leq k$ let $\aO_i$ be such that $(\tr,\delta)(q_1,\aO_i,\AT) = E_i$. So, for all delayed $q_1$-response $\overline{e}$ and all $\aO \in \AO$, there exists $i \leq k$ such that $(\tr',\delta')(q_1,\aO_i,\AT) = \overline{e}_i= (\tr',\delta')(q_1,\aO,\AT)$ (where $g_{q_0}(q_1,\overline{e}) = \langle  \AO, \AT, Q, q, \delta’, \Col, \tr’, W\rangle$). Informally, whichever simpler game $g_{q_0}(q_1,\overline{e})$ is being played, for any action  $\aO \in \AO$ at state $q_1$, there is some $\aO_i$ that has the same effect. So wlog let us assume that the $\sO_{\overline{e}}$ only prescribe actions in $\{\aO_1,\dots,\aO_k\}$ when at state $q_1$.

Let $M$ be a set with cardinality $|M| = 2^r$, let and $m_0 \in M$. By assumption, for all $\overline{e} \in E_1 \times \dots \times E_k$ the strategy $\sO_{\overline{e}}$ can be represented as follows: $\sigma_{\overline{e}}: Q \times M \to \AO$ and $\mu_{\overline{e}}: Q \times M \times \AT \to M$ such that $\sO_{\overline{e}}(\aHT) = \sigma_{\overline{e}}(q_{\overline{e}}(\aHT),m_{\overline{e}}(\aHT))$, where $m_{\overline{e}}$ and $q_{\overline{e}}$ are defined by mutual induction below. 
\begin{itemize}
\item $m_{\overline{e}}(\e) := m_0$ and $q_{\overline{e}^1}(\e) := q_0$ and  $q_{\overline{e}}(\e) := q_1$ for all $\overline{e} \neq \overline{e}^1$

\item $m_{\overline{e}}(\aHT\aT) := \mu_{\overline{e}}(q_{\overline{e}}(\aHT),m_{\overline{e}}(\aHT),\aT)$ and $q_{\overline{e}}(\aHT\aT) := \delta(q_{\overline{e}}(\aHT),\sigma_{\overline{e}}(q_{\overline{e}}(\aHT),m_{\overline{e}}(\aHT)),\aT)$
\end{itemize}

By Corollary~\ref{cor:fpi} $\exists F: (E_1 \times \dots \times E_k \to \{1,\dots,k\}) \to \{1,\dots,k\}, \forall f : E_1 \times \dots \times E_k \to \{1,\dots, k\} , \exists G_f : E_{F(f)} \to E_1 \times \dots \times E_k, \forall e \in E_{F(f)}, G_f(e)_{F(f)} = e \wedge f \circ G_f(e) = F(f)$.

A \pO\/ strategy is built as follows.
\begin{itemize}
\item ${\bf M} := \prod_iE_i \cup \{G\} \to M \cup \prod_iE_i$ (or $M^{\prod_iE_i} \times \prod_iE_i$)
\item ${\bf m_0} \in {\bf M}$ is defined by ${\bf m_0}(\overline{e}) := m_0$ for all $\overline{e} \in E_1 \times \dots \times E_k$ and ${\bf m_0}(G) := \overline{e}^1$

\item 
\begin{itemize}
\item $\iota:\{\aO_1,\dots,\aO_k\} \to \{1,\dots,k\}$ is defined by $\iota(\aO_i) := i$ for all $i \in \{1,\dots,k\}$.

\item $f_m: E_1 \times \dots \times E_k \to \{1,\dots, k\}$ is defined by $f_m(\overline{e})  :=\iota \circ \sigma_{\overline{e}} (q_1,m(\overline{e}))$ for all $m \in {\bf M}$.

\item $\overline{d}_{m,\aT} := G_{f_m}\circ (\tr,\delta)(q_1,\aO_{F(f_m)}, \aT) \in E_{F(f_m)}$ for all $(m,\aT) \in {\bf M} \times \AT$.
\end{itemize}

\item $\boldsymbol{\sigma}: Q \times {\bf M} \to \AO$ is defined by
$\boldsymbol{\sigma}(q,m) :=
\begin{cases}
\aO_{F(f_m)} \mbox{ if }q = q_1\\
\sigma_{m(G)}(q,m \circ m(G)) \mbox{ otherwise}
\end{cases}
$

\item $\boldsymbol{\mu}: Q \times {\bf M} \times \AT \to {\bf M}$ is defined as follows. Let $(q,m,\aT,\overline{e}) \in Q \times {\bf M} \times \AT \times ( E_1 \times \dots \times E_k)$. 

$
\boldsymbol{\mu}(q,m,\aT)(G) :=
\begin{cases}
\overline{d}_{m,\aT} \mbox{ if }q = q_1,\\
m(G) \mbox{ otherwise.}
\end{cases}
$

$
\boldsymbol{\mu}(q,m,\aT)(\overline{e}) := 
\begin{cases}
\mu_{\overline{e}}(q,m(\overline{e}),\aT) \mbox{ if }\boldsymbol{\mu}(q,m,\aT)(G) = \overline{e},\\
m(\overline{e}) \mbox{ otherwise.}
\end{cases}
$

\item 
\begin{itemize}
\item ${\bf m}(\e) := {\bf m_0}$ and ${\bf q}(\e) := q_0$
\item ${\bf m}(\aHT \aT) := \boldsymbol{\mu}({\bf q}(\aHT),{\bf m}(\aHT),\aT)$ and ${\bf q}(\aHT \aT) := \delta({\bf q}(\aHT),\boldsymbol{\sigma}({\bf q}(\aHT),{\bf m}(\aHT)),\aT)$
\end{itemize}
%For all $b \in \AT$, let $\overline{d} := G_{f}(\tr(\aO_{F(f)}, \aT))$ and let $\upsilon(m_{\overline{\col}})_{\overline{\col} \in \Col_1 \times \dots \times \Col_k}, \aT) := (m'_{\overline{\col}})_{\overline{\col} \in \Col_1 \times \dots \times \Col_k}$ where $m'_{\overline{\col}} := m_{\overline{\col}}$ for all $\overline{\col} \neq \overline{d}$ and $m'_{\overline{d}} := \upsilon_{\overline{\col}}(m_{\overline{d}})$.
\end{itemize}

Note that $\log_2 |{\bf M}| = \log_2(|M|^dd) = d \log_22^m + \log_2d = dm + \log_2 d$, where $d = \prod_{i=1}^k|E_i|$.

Claim~\ref{claim1} below says that if the current “simpler game” has just changed, the previous state must have been $q_1$. Claim~\ref{claim2} below says that if the previous state was $q_1$, the current “simpler game” can be expressed using $\overline{d}$. Claim~\ref{claim2} below says that if the current state was $q_1$, the action that \pO\/ has just chosen in the compound game is the same as the action that \pO\/ would have chosen in the current “simpler game” under similar memory environment. By inspecting the above definitions of ${\bf m}$ and $\boldsymbol{\mu}$, one can show Claims~\ref{claim1} and \ref{claim2} below for all $\aHT \in \AT^*$ and $\aT \in \AT$.
\begin{align}
{\bf m}(\aHT \aT)(G) \neq{\bf m}(\aHT)(G) &\,\Rightarrow\, {\bf q}(\aHT) = q_1 \label{claim1}\\
{\bf q}(\aHT) = q_1 &\,\Rightarrow\,{\bf m}(\aHT \aT)(G) = \overline{d}_{{\bf m}(\aHT), \aT}\label{claim2}\\
{\bf q}(\aHT) = q_1 &\,\Rightarrow\, \boldsymbol{\sigma}(q_1,{\bf m}(\aHT)) = \sigma_{{\bf m}(\aHT \aT)(G)}(q_1,{\bf m}(\aHT)({\bf m}(\aHT \aT)(G)))\label{claim3}
\end{align}
Claim~\ref{claim3} requires a proof: for all $\aHT \in \AT^*$ and $\aT \in \AT$
\begin{align*}
\iota \circ \boldsymbol{\sigma}(q_1,{\bf m}(\aHT)) &= F(f_{{\bf m}(\aHT)}) \mbox{ by definition of }\boldsymbol{\sigma},\\
	& = f_{{\bf m}(\aHT)} \circ G_{f_{{\bf m}(\aHT)}} \circ (\tr,\delta)(q_1,\aO_{F(f_{{\bf m}(\aHT)})}, \aT)\mbox{ by definition of }G_{f_{{\bf m}(\aHT)}},\\
	& = f_{{\bf m}(\aHT)}\circ \boldsymbol{\mu}(q_1,{\bf m}(\aHT),\aT)(G) \mbox{ by definition of }\boldsymbol{\mu}\\
	& = f_{{\bf m}(\aHT)} \circ \boldsymbol{\mu}({\bf q}(\aHT),{\bf m}(\aHT),\aT)(G) \mbox{ since } {\bf q}(\aHT) = q_1 \mbox{ by assumption,}\\
	& = f_{{\bf m}(\aHT)} \circ {\bf m}(\aHT \aT)(G) \mbox{ by definition of }{\bf m},\\
	& = f_{{\bf m}(\aHT)} (\overline{d}) \mbox{ by Claim~\ref{claim2}, and setting }\overline{d} := \overline{d}_{{\bf m}(\aHT), \aT},\\	
	& = \iota \circ  \sigma_{\overline{d}}(q_1,{\bf m}(\aHT)(\overline{d}))\mbox{ by definition of }f_{{\bf m}(\aHT)}.
\end{align*}
Therefore $\boldsymbol{\sigma}({\bf q}(\aHT),{\bf m}(\aHT)) = \sigma_{\overline{d}}({\bf q}(\aHT),{\bf m}(\aHT)(\overline{d}))$ since $\iota$ is injective and ${\bf q}(\aHT) = q_1$ by assumption, and Claim~\ref{claim3} is proved.

By definition, proving that the \pO\/ strategy $\boldsymbol{\sigma} \circ ({\bf q},{\bf m})$ is winning amounts to proving that $\tr \circ \hOS(\boldsymbol{\sigma} \circ({\bf q},{\bf m}) ,\aHTinf) \in W$ for all \pT\/ runs $\aHTinf \in \AT^\omega$. This will be done by proving that all $\aHT \in \AT^*$ can be decomposed into $(\aHT^{\overline{e}})_{\overline{e} \in E_1 \times \dots \times E_k}$ such that $\tr \circ \hOS(\boldsymbol{\sigma} \circ({\bf q},{\bf m}) ,\aHT)$ can be obtained by interleaving the $\tr \circ \hOS(\sigma_{\overline{e}} circ(q_{\overline{e}},m_{\overline{e}}) ,\aHT^{\overline{e}})$. The decomposition of the \pT\/ history is done by induction: for all $\aHT \in \AT^*$, all $\aT \in \AT$, and all $\overline{e} \in (E_1 \times \dots \times E_k)$
$$\begin{array}{l}
\e^{\overline{e}} := \e\\
(\aHT \aT)^{{\bf m}(\aHT \aT)(G)} := \aHT^{{\bf m}(\aHT \aT)(G)} \aT\\
(\aHT \aT)^{\overline{e}} :=\aHT^{\overline{e}} \mbox{ if }\overline{e} \neq {\bf m}(\aHT \aT)(G)
\end{array}$$
 
Claim~\ref{claim4} says that the compound memory stores the memory contents related to each simpler game.
Claim~\ref{claim5} says that the current state in the compound game is equal to the current state in the current simpler game.
Claim~\ref{claim6} says that the current state in the simpler games on stand-by is $q_1$.
$\forall \aHT \in \AT^*,\forall \overline{e} \in  E_1 \times \dots \times E_k$
\begin{align}
{\bf m}(\aHT)( \overline{e}) = m_{\overline{e}}(\aHT^{\overline{e}})\label{claim4}\\
q_{{\bf m}(\aHT)(G)}(\aHT^{{\bf m}(\aHT)(G)}) = {\bf q}(\aHT)\label{claim5}\\
 \overline{e} \neq {\bf m}(\aHT)(G)\,\Rightarrow\,q_{\overline{e}}(\aHT^{\overline{e}}) = q_1\label{claim6}
\end{align} 
Claims~\ref{claim4},\ref{claim5},\ref{claim6} are proved by mutual induction on $\aHT$.  

Base case, $\aHT = \e$.
$$\begin{array}{l}
{\bf m}(\e)( \overline{e}) =  {\bf m_0}( \overline{e}) = m_0 = m_{\overline{e}}(\e) = m_{\overline{e}}(\e^{\overline{e}}).\\
q_{{\bf m}(\e)(G)}(\e^{{\bf m}(\e)(G)}) = q_{{\bf m_0}(G)}(\e) = q_{\overline{e}^1}(\e) = q_0 ={\bf q}(\e).\\
\mbox{If }\overline{e} \neq {\bf m}(\e)(G)\mbox{ then }\overline{e} \neq {\bf m_0}(G) =  \overline{e}^1\mbox{ and }q_{\overline{e}}(\e^{\overline{e}}) =  q_{\overline{e}}(\e)= q_1.
\end{array}$$

%Note that ${\bf m}(\aHT \aT)(G) = \boldsymbol{\mu}({\bf q}(\aHT),{\bf m}(\aHT),\aT)(G)$ and ${\bf m}(\aHT \aT)(\overline{e}) = \boldsymbol{\mu}({\bf q}(\aHT),{\bf m}(\aHT),\aT)(\overline{e})$ for all $\overline{e}$, since ${\bf m}(\aHT \aT) := \boldsymbol{\mu}({\bf q}(\aHT),{\bf m}(\aHT),\aT)$ by definition of ${\bf m}$.

For the inductive case let  $\aHT \in \AT^*$, $\aT \in \AT$, and $\overline{e} \in E_1 \times \dots \times E_k$ . To prove Claim~\ref{claim4}, let us make a case disjunction. First case, $\overline{e} \neq \boldsymbol{\mu}({\bf q}(\aHT),{\bf m}(\aHT),\aT)(G)$, so $\overline{e} \neq {\bf m}(\aHT \aT)(G)$ by definition of ${\bf m}$, so $\aHT^{\overline{e}} = (\aHT \aT)^{\overline{e}}$. Then
\begin{align*}
{\bf m}(\aHT \aT)(\overline{e}) & = {\bf m}(\aHT)(\overline{e}) \mbox{ by definition of }\boldsymbol{\mu},\\
	& = m_{\overline{e}}(\aHT^{\overline{e}}) \mbox{ by I.H.,}\\
	& = m_{\overline{e}}((\aHT \aT)^{\overline{e}}) \mbox{ since } \aHT^{\overline{e}} = (\aHT \aT)^{\overline{e}} \mbox{ as mentioned above.}
\end{align*}
Second case, $\overline{e} = \boldsymbol{\mu}({\bf q}(\aHT),{\bf m}(\aHT),\aT)(G)$, so $\overline{e} = {\bf m}(\aHT \aT)(G)$ and $\aHT^{\overline{e}}b = (\aHT \aT)^{\overline{e}}$. So
\begin{align*}
{\bf m}(\aHT \aT)(\overline{e}) & =  \mu_{\overline{e}}({\bf q}(\aHT),{\bf m}(\aHT)(\overline{e}),\aT)\mbox{ by definition of }\boldsymbol{\mu},\\
	& =  \mu_{\overline{e}}({\bf q}(\aHT),m_{\overline{e}}(\aHT^{\overline{e}}),\aT)\mbox{  since }{\bf m}(\aHT)(\overline{e}) = m_{\overline{e}}(\aHT^{\overline{e}}) \mbox{ by I.H.}
\end{align*}
Let us make a further case disjunction. First sub-case, $\overline{e} = {\bf m}(\aHT)(G)$. Then
\begin{align*}
{\bf m}(\aHT \aT)(\overline{e}) & =  \mu_{\overline{e}}(q_{\overline{e}}(\aHT^{\overline{e}}),m_{\overline{e}}(\aHT^{\overline{e}}),\aT)\mbox{  since }{\bf q}(\aHT) = q_{\overline{e}}(\aHT^{\overline{e}}) \mbox{ by I.H. (and }\overline{e} = {\bf m}(\aHT)(G) \mbox{),}\\
	& = m_{\overline{e}}(\aHT^{\overline{e}} \aT) \mbox{ by definition of }m_{\overline{e}},\\
	& = m_{\overline{e}}((\aHT \aT)^{\overline{e}}) \mbox{ since } \aHT^{\overline{e}} \aT= (\aHT \aT)^{\overline{e}} \mbox{ as mentioned above.}
\end{align*}
Second sub-case, $\overline{e} \neq {\bf m}(\aHT)(G)$, so ${\bf q}(\aHT) = q_1$ by Claim~\ref{claim1}, and $\overline{e} = \overline{d}$ by Claim~\ref{claim2}. Then
\begin{align*}
{\bf m}(\aHT \aT)(\overline{e}) & =  \mu_{\overline{e}}(q_1,m_{\overline{e}}(\aHT^{\overline{e}}),\aT)\mbox{  since }{\bf q}(\aHT) = q_1 \mbox{ as mentioned above,}\\
	& = \mu_{\overline{e}}(q_{\overline{e}}(\aHT^{\overline{e}}),m_{\overline{e}}(\aHT^{\overline{e}}),\aT)\mbox{  since } q_{\overline{e}}(\aHT^{\overline{e}})= q_1\mbox{ by I.H. (and }\overline{e} \neq {\bf m}(\aHT)(G) \mbox{),}\\
	& = m_{\overline{e}}(\aHT^{\overline{e}} \aT) \mbox{ by definition of }m_{\overline{e}},\\
	& = m_{\overline{e}}((\aHT \aT)^{\overline{e}}) \mbox{ since } \aHT^{\overline{e}} \aT= (\aHT \aT)^{\overline{e}} \mbox{ as mentioned above.}
\end{align*}

{\bf Claim~\ref{claim5}.}
Let $\overline{e} := {\bf m}(\aHT \aT)(G)$. To prove that $q_{\overline{e}}((\aHT \aT)^{\overline{e}}) = {\bf q}(\aHT \aT)$, let us also make a case disjunction. First case, ${\bf q}(\aHT) \neq q_1$, so ${\bf m}(\aHT \aT)(G) = {\bf m}(\aHT)(G)$ by Claim~\ref{claim1}.
\begin{align*}
q_{\overline{e}}((\aHT \aT)^{\overline{e}}) & = q_{\overline{e}}(\aHT^{\overline{e}}\aT) \mbox{ since }\overline{e} = {\bf m}(\aHT \aT)(G) \mbox{ implies }(\aHT \aT)^{\overline{e}} = \aHT^{\overline{e}}\aT,\\
	& = \delta(q_{\overline{e}}(\aHT^{\overline{e}}), \sigma_{\overline{e}}(q_{\overline{e}} (\aHT^{\overline{e}}), m_{\overline{e}}(\aHT^{\overline{e}})), \aT) \mbox{ by definition of }q_{\overline{e}},\\
		& = \delta({\bf q}(\aHT), \sigma_{\overline{e}}({\bf q}(\aHT), m_{\overline{e}}(\aHT^{\overline{e}})), \aT) \mbox{ since }q_{\overline{e}}(\aHT^{\overline{e}}) = {\bf q}(\aHT) \mbox{ by I.H.,}\\
		& =  \delta({\bf q}(\aHT), \sigma_{\overline{e}}({\bf q}(\aHT), {\bf m}(\aHT)(\overline{e}), \aT) \mbox{ since }m_{\overline{e}}(\aHT^{\overline{e}}) = {\bf m}(\aHT)(\overline{e}) \mbox{ by I.H.,}\\
		& =  \delta({\bf q}(\aHT), \sigma_{{\bf m}(\aHT)(G)}({\bf q}(\aHT), {\bf m}(\aHT)({\bf m}(\aHT)(G)), \aT) \mbox{ since }\overline{e} = {\bf m}(\aHT)(G) \mbox{ as argued above,}\\
		& =  \delta({\bf q}(\aHT),\boldsymbol{\sigma}({\bf q}(\aHT),{\bf m}(\aHT)) ,\aT)  \mbox{ by definition of }\boldsymbol{\sigma}\mbox{ since }{\bf q}(\aHT) \neq q_1 \mbox{ by assumption,}\\
		& = {\bf q}(\aHT \aT) \mbox{ by definition of }{\bf q}.
\end{align*}
Second case, ${\bf q}(\aHT) = q_1$, so
\begin{align*}
{\bf q}(\aHT \aT) & = \delta({\bf q}(\aHT),\boldsymbol{\sigma}({\bf q}(\aHT),{\bf m}(\aHT)), \aT)\mbox{ by definition of }{\bf q},\\
	& = \delta({\bf q}(\aHT), \sigma_{\overline{e}}({\bf q}(\aHT),{\bf m}(\aHT)(\overline{e})), \aT)\mbox{ by Claim~\ref{claim3},}\\
	& = \delta({\bf q}(\aHT), \sigma_{\overline{e}}({\bf q}(\aHT),m_{\overline{e}}(\aHT^{\overline{e}})), \aT) \mbox{ since }{\bf m}(\aHT)(\overline{e}) =m_{\overline{e}}(\aHT^{\overline{e}}) \mbox{ by I.H,}\\
	& = \delta(q_{\overline{e}}(\aHT^{\overline{e}}), \sigma_{\overline{e}}(q_{\overline{e}}(\aHT^{\overline{e}}),m_{\overline{e}}(\aHT^{\overline{e}})), \aT) \mbox{ by I.H. (whether or not }\overline{e} = {\bf m}(\aHT)(G)\mbox{),}\\
	& = q_{\overline{e}}(\aHT^{\overline{e}} \aT) \mbox{ by definition of }q_{\overline{e}},\\
	& = q_{\overline{e}}((\aHT \aT)^{\overline{e}}) \mbox{ since } \overline{e} = {\bf m}(\aHT \aT)(G).
\end{align*}

{\bf Claim~\ref{claim6}.}
To prove that $\overline{e} \neq {\bf m}(\aHT \aT)(G)\,\Rightarrow\,q_{\overline{e}}((\aHT \aT)^{\overline{e}}) = q_1$, let $\overline{e} \neq {\bf m}(\aHT \aT)(G)$. If $\overline{e} \neq {\bf m}(\aHT)(G)$ then $q_{\overline{e}}(\aHT^{\overline{e}}) = q_1$ by I.H.; if $\overline{e} = {\bf m}(\aHT)(G)$ then ${\bf m}(\aHT)(G) \neq {\bf m}(\aHT \aT)(G)$, and ${\bf q}(\aHT) = q_1$ by Claim~\ref{claim1}, so $q_{\overline{e}}(\aHT^{\overline{e}}) = q_1$ by I.H. Now, $q_{\overline{e}}((\aHT \aT)^{\overline{e}}) = q_{\overline{e}}(\aHT ^{\overline{e}})$ since $\overline{e} \neq {\bf m}(\aHT \aT)(G)$, so $q_{\overline{e}}((\aHT \aT)^{\overline{e}}) = q_1$.

Claims~\ref{claim4},\ref{claim5},\ref{claim6} are thus proved. Let us now prove that for all $\aHT \in \AT^*$ and all $\aT \in \AT$
\begin{align}
\boldsymbol{\sigma}({\bf q}(\aHT),{\bf m}(\aHT)) = \sigma_{\overline{e}}(q_{\overline{e}}(\aHT^{\overline{e}}),m_{\overline{e}}(\aHT^{\overline{e}})) \mbox{ where }\overline{e} := {\bf m}(\aHT \aT)(G)\label{claim7}
\end{align}
Let us make a case disjunction. First case, ${\bf q}(\aHT) = q_1$, so
\begin{align*}
\boldsymbol{\sigma}({\bf q}(\aHT),{\bf m}(\aHT)) & =\boldsymbol{\sigma}(q_1,{\bf m}(\aHT)),\\
	& = \sigma_{\overline{e}}(q_1, {\bf m}(\aHT)(\overline{e})) \mbox{ by Claim~\ref{claim3},}\\
	& = \sigma_{\overline{e}}(q_1, m_{\overline{e}}(\aHT^{\overline{e}})) \mbox{ by Claim~\ref{claim4},}\\
	& = \sigma_{\overline{e}}(q_{\overline{e}}(\aHT^{\overline{e}}), m_{\overline{e}}(\aHT^{\overline{e}}))\mbox{ by Claims~\ref{claim5} and \ref{claim6} and since }{\bf q}(\aHT) = q_1.
\end{align*}
Second case, ${\bf q}(\aHT) \neq q_1$, so
\begin{align*}
\boldsymbol{\sigma}({\bf q}(\aHT),{\bf m}(\aHT)) & =\sigma_{{\bf m}(\aHT)(G)}({\bf q}(\aHT),{\bf m}(\aHT) \circ {\bf m}(\aHT)(G)) \mbox{ by definition of }\boldsymbol{\sigma},\\
	& = \sigma_{\overline{e}}({\bf q}(\aHT),{\bf m}(\aHT)(\overline{e})) \mbox{ since }\overline{e} = {\bf m}(\aHT)(G)\mbox{ by Claim~\ref{claim1},}\\
	& = \sigma_{\overline{e}}(q_{\overline{e}}(\aHT^{\overline{e}}),{\bf m}(\aHT)(\overline{e})) \mbox{ by Claim~\ref{claim5} since }\overline{e} = {\bf m}(\aHT)(G),\\
	& = \sigma_{\overline{e}}(q_{\overline{e}}(\aHT^{\overline{e}}),m_{\overline{e}}(\aHT^{\overline{e}})) \mbox{ by Claim~\ref{claim4}.}
\end{align*}
This proves Claim~\ref{claim7}. Let us prove the last intermediate claim that for all $\aHT \in \AT^*$, $\tr \circ \hOS(\boldsymbol{\sigma} \circ ({\bf q},{\bf m}),\aHT)$ can be obtained by interleaving the  $\tr \circ \hOS(\sigma_{\overline{e}} \circ (q_{\overline{e}},m_{\overline{e}}),\aHT^{\overline{e}})$. Let us proceed  by induction on $\aHT$. For the base case $\e$ can indeed be obtained by interleaving several $\e$. For the inductive case, let $\aHT \in \AT^*$ and $\aT \in \AT$.
\begin{align*}
\tr \circ \hOS(\boldsymbol{\sigma} \circ ({\bf q},{\bf m}),\aHT \aT) & = \tr ( \hOS(\boldsymbol{\sigma} \circ ({\bf q},{\bf m}),\aHT) \cdot (\boldsymbol{\sigma} \circ ({\bf q},{\bf m})(\aHT), \aT)) \mbox{ by definition of }\hOS,\\
	& =  \tr \circ \hOS(\boldsymbol{\sigma} \circ ({\bf q},{\bf m}),\aHT) \cdot \tr(\boldsymbol{\sigma} \circ ({\bf q},{\bf m})(\aHT), \aT)\mbox{ by definition of }\tr,\\
	& =  \tr \circ \hOS(\boldsymbol{\sigma} \circ ({\bf q},{\bf m}),\aHT) \cdot \tr(\sigma_{\overline{e}^0}(q_{\overline{e}^0}(\aHT^{\overline{e}^0}),m_{\overline{e}^0}(\aHT^{\overline{e}^0})), \aT),\\
	& \qquad\mbox{ where }\overline{e}^0 := {\bf m}(\aHT \aT)(G) \mbox{ by Claim~\ref{claim7}.}
\end{align*}
By I.H., $\tr \circ \hOS(\boldsymbol{\sigma} \circ ({\bf q},{\bf m}),\aHT)$ can be obtained by interleaving the $\tr \circ \hOS(\sigma_{\overline{e}} \circ (q_{\overline{e}},m_{\overline{e}}),\aHT^{\overline{e}})$. By definition $\aHT^{\overline{e}} = (\aHT\aT)^{\overline{e}}$ for all $\overline{e} \neq \overline{e}^0$. So $\tr \circ \hOS(\boldsymbol{\sigma} \circ ({\bf q},{\bf m}),\aHT)$ can be obtained by interleaving of $\tr \circ \hOS(\sigma_{\overline{e}^0} \circ (q_{\overline{e}^0},m_{\overline{e}^0}),\aHT^{\overline{e}^0})$ and the $\tr \circ \hOS(\sigma_{\overline{e}} \circ (q_{\overline{e}},m_{\overline{e}}),(\aHT\aT)^{\overline{e}})$ for $\overline{e} \neq \overline{e}^0$. But 
\begin{align*}
\tr \circ \hOS(\sigma_{\overline{e}^0} \circ (q_{\overline{e}^0},m_{\overline{e}^0}),(\aHT\aT)^{\overline{e}^0}) & = \tr \circ \hOS(\sigma_{\overline{e}^0} \circ (q_{\overline{e}^0},m_{\overline{e}^0}),\aHT^{\overline{e}^0}\aT) \mbox{ by definition of }\aHT \mapsto \aHT^{\overline{e}},\\
	& = \tr \circ \hOS(\sigma_{\overline{e}^0} \circ (q_{\overline{e}^0},m_{\overline{e}^0}),\aHT^{\overline{e}^0}) \cdot \tr(\sigma_{\overline{e}^0}(q_{\overline{e}^0}(\aHT^{\overline{e}^0}),m_{\overline{e}^0}(\aHT^{\overline{e}^0})), \aT),
\end{align*}
so by the definition of interleaving, $\tr \circ \hOS(\boldsymbol{\sigma} \circ ({\bf q},{\bf m}),\aHT\aT)$ can be obtained by interleaving of the $\tr \circ \hOS(\sigma_{\overline{e}} \circ (q_{\overline{e}},m_{\overline{e}}),(\aHT\aT)^{\overline{e}})$, thus completing the induction.

Finally, let $\aHTinf \in \AT^\omega$. The definition of $\aHT\mapsto \aHT^{\overline{e}}$ extends naturally to infinite arguments, and so does the last claim by Lemma~\ref{lem:interleaving-eq}. Said otherwise, for all $\aHTinf \in \AT^\omega$, $\tr \circ \hOS(\boldsymbol{\sigma} \circ ({\bf q},{\bf m}),\aHTinf)$ can be obtained by interleaving of the $\tr \circ \hOS(\sigma_{\overline{e}} \circ (q_{\overline{e}},m_{\overline{e}}),\aHTinf^{\overline{e}}) \in W$. By interleaving closeness $\tr \circ \hOS(\boldsymbol{\sigma} \circ ({\bf q},{\bf m}),\aHTinf) \in W$, so  $\boldsymbol{\sigma} \circ ({\bf q},{\bf m})$ is a \pO\/ winning strategy in $g$.
\end{proof}
\end{lemma}

\begin{observation}\label{obs:cg-involve-less}
For every game $g = \langle  \AO, \AT, Q, q_0, \delta, \Col, \tr, W\rangle$, state $q \in Q$  involving \pT\/, and delayed $q$-response $\overline{e}$, the following holds.
\begin{enumerate}
\item\label{obs:cg-involve-less1} \pT\/ involves one state less (i.e. $q$) in $g(q,\overline{e})$ than in $g$.
\item\label{obs:cg-involve-less2} Every delayed response $t$ in $g(q,\overline{e})$ is also a delayed response in $g$, and $t_q = \overline{e}$ and $g(q,\overline{e})(t) = g(t)$.
\end{enumerate}
\begin{proof} 
\begin{enumerate}
\item
With $g(q,\overline{e}) := \langle  \AO, \AT, Q, q_0, \delta', \Col, \tr', W\rangle$, one finds $|(\tr,\delta)(q',\aO,\AT)| = |(\tr',\delta')(q',\aO,\AT)|$ for all $(q',\aO) \in (Q\setminus \{q\}) \times \AO$, and $(\tr',\delta')(q,\aO,\AT) = \{\overline{e}_{(\tr,\delta)(q,\aO,\AT)}\}$ so $|(\tr',\delta')(q,\aO,\AT)| = 1$.
\end{enumerate}
\end{proof}
\end{observation}

Lemma~\ref{lem:sub-strat-trace} below shows that the trace behaves nicely w.r.t. the “sub-strategies”.
\begin{lemma}\label{lem:sub-strat-trace}
Consider games $g_{q} = \langle  \AO, \AT, Q, q, \delta, \Col, \tr, W\rangle$ parametrized by $q \in Q$. Let $q_0 \in Q$, let $\sO$ be a \pO\/ strategy, let $\aHT' \in \AT^*$, let $q := \Delta_{q_0} \circ \hOS(\sO,\aHT')$, and let a \pO\/ strategy $\sO_q$ be defined by $\sO_q(\aHT) := \sO(\aHT^q \aHT)$ for all $\aHT \in \AT^*$. Then 
$$\forall\aHTinf \in \AT^\omega,\,\tr_{q_0} \circ \hOS(\sO,\aHT^q) \cdot \tr_q \circ \hOS(\sO_q,\aHTinf) = \tr_{q_0} \circ \hOS(\sO,\aHT^q\aHTinf)$$
and if $W$ is closed under prefix removal and if $\sO$ wins $g_{q_0}$, then $\sO_q$ wins $g_{q}$. Moreover, if $\sO$ uses memory $M$, so does $\sO_q$.

\begin{proof}
Let us prove the following by induction on $\aHT$. (where $\Delta_{q}(\e) = q$)
\begin{align}
\forall \aHT \in \AT^*,\,\Delta_q \circ \hOS(\sO_q,\aHT) = \Delta_{q_0} \circ \hOS(\sO,\aHT^q\aHT)\label{eq:prefix removal-win-one-eq4}
\end{align}
For the base case, $\Delta_q \circ \hOS(\sO_q,\e) = q = \Delta_{q_0} \circ \hOS(\sO,\aHT^q)$. For the inductive case let $\aHT \in \AT^*$ and $\aT \in \AT$. Then
\begin{align*}
\Delta_q \circ \hOS(\sO_q,\aHT \aT) & = \Delta_q(\hOS(\sO_q,\aHT) \cdot (\sO_q(\aHT),\aT)) \mbox{ by definition of }\hOS,\\
	& = \delta(\Delta_q \circ \hOS(\sO_q,\aHT), \sO_q(\aHT),\aT))\mbox{ by definition of }\Delta_q,\\
	& =  \delta(\Delta_{q_0} \circ \hOS(\sO,\aHT^q\aHT), \sO_q(\aHT),\aT))\mbox{ by I.H.,}\\
	& =  \delta(\Delta_{q_0} \circ \hOS(\sO,\aHT^q\aHT), \sO(\aHT^q\aHT),\aT))\mbox{ by definition of }\sO_q,\\
	& = \Delta_{q_0} (\hOS(\sO,\aHT^q\aHT) \cdot \sO(\aHT^q\aHT),\aT)\mbox{ by definition of }\Delta_{q_0},\\
	& = \Delta_{q_0} \circ \hOS(\sO,\aHT^q\aHT \aT)  \mbox{ by definition of }\hOS\mbox{, thus completing the induction.}
\end{align*}

Let us now prove the following by induction on $\aHT$.
\begin{align}
\forall\aHT \in \AT^*,\,\tr_{q_0} \circ \hOS(\sO,\aHT^q) \cdot \tr_q \circ \hOS(\sO_q,\aHT) = \tr_{q_0} \circ \hOS(\sO,\aHT^q\aHT)\label{eq:prefix removal-win-one-eq5}
\end{align}
For the base case, $\tr_{q_0} \circ \hOS(\sO,\aHT^q) \cdot \tr_q \circ \hOS(\sO_q,\e) = \tr_{q_0} \circ \hOS(\sO,\aHT^q) \cdot \e = \tr_{q_0} \circ \hOS(\sO,\aHT^q\e)$. For the inductive case, let $\aHT \in \AT^*$ and $\aT \in \AT$. Then $\tr_{q_0} \circ \hOS(\sO,\aHT^q) \cdot \tr_q \circ \hOS(\sO_q,\aHT \aT)$ 
\begin{align*}
& = \tr_{q_0} \circ \hOS(\sO,\aHT^q) \cdot\tr_q(\hOS(\sO_q,\aHT) \cdot (\sO_q(\aHT),\aT)) \mbox{ by definition of }\hOS,\\
	& = \tr_{q_0} \circ \hOS(\sO,\aHT^q) \cdot \tr_q \circ \hOS(\sO_q,\aHT) \cdot \tr(\Delta_q \circ \hOS(\sO_q,\aHT),\sO_q(\aHT),\aT ) \mbox{ by definition of }\tr_q,\\
	& = \tr_{q_0} \circ \hOS(\sO,\aHT^q\aHT) \cdot \tr(\Delta_q \circ \hOS(\sO_q,\aHT),\sO_q(\aHT),\aT)  \mbox{ by I.H.,}\\
	& = \tr_{q_0} \circ \hOS(\sO,\aHT^q\aHT) \cdot \tr(\Delta \circ \hOS(\sO,\aHT^q\aHT),\sO_q(\aHT),\aT)\mbox{ by Claim~(\ref{eq:prefix removal-win-one-eq4}),}\\
	& = \tr_{q_0} \circ \hOS(\sO,\aHT^q\aHT) \cdot \tr(\Delta \circ \hOS(\sO,\aHT^q\aHT),\sO(\aHT^q\aHT),\aT)\mbox{ by definition of }\sO_q,\\
	& = \tr_{q_0}(\hOS(\sO,\aHT^q\aHT) \cdot (\sO(\aHT^q\aHT),\aT))\mbox{ by definition of }\tr_r,\\
	& = \tr_{q_0} \circ \hOS(\sO,\aHT^q\aHT \aT)\mbox{ by definition of }\hOS\mbox{, thus completing the induction.}
\end{align*}
Claim~(\ref{eq:prefix removal-win-one-eq5}) is lifted to infinite arguments: $\forall\aHTinf \in \AT^\omega,\,\tr_q \circ \hOS(\sO_q,\aHTinf) = \tr \circ \hOS(\sO,\aHT^q\aHTinf)$. Since $\sO$ is winning in the original game, $\tr_q \circ \hOS(\sO_q,\AT^\omega) = tr \circ \hOS(\sO,\aHT^q\AT^\omega) \subseteq W$, so $\sO_q$ is winning $g_{q}$.

Moreover, if $\sO$'s being finite memory is witnessed by $m_0$, $\sigma$, and $\mu$, then $\sO_q$'s being finite memory is witnessed by $m(\aHT^q)$, $\sigma$, and $\mu$.
\end{proof}
\end{lemma}

\begin{proof}[Proof of Theorem~\ref{thm:interleaving-prefix removal-eq}]
\ref{thm:interleaving-win-all-eq1} $\Rightarrow$\ref{thm:interleaving-win-all-eq2} by Lemma~\ref{lem:win-implies-win-delayed-response}.

{\bf Sketch of \ref{thm:interleaving-win-all-eq2} $\Rightarrow$\ref{thm:interleaving-win-all-eq1}}
The proof proceeds by induction on the number of states that involve \pT\/. For the base case, morally ``$g_{q_0} = g_{q_0}(t)$'' for all delayed responses $t$. For the inductive case, let a state $q_1$ involve \pT\/. In the games $g_{q_0}(q_1,\overline{e})$, where $\overline{e}$ is a delayed $q_1$-response, fewer states involve \pT\/, so the I.H. yields corresponding winning strategies. If one of these strategies avoids $q_1$, it wins also in $g_{q_0}$. Otherwise, each of these strategies has a ``substrategy'' that wins $g_{q_1}(q_1,\overline{e})$, by closeness under prefix removal.  Lemma~\ref{lem:win-delayed-response-implies-win} is invoked to complete the induction.

{\bf Details of \ref{thm:interleaving-win-all-eq2} $\Rightarrow$\ref{thm:interleaving-win-all-eq1}}
The proof proceeds by induction on the number of states that involve \pT\/. For the base case, let us assume that no state involves \pT\/. Especially, $(\tr,\delta)(q,\aO, \AT) = \{t_{q,(\tr,\delta)(q,\aO,\AT)}\}$ for all $(q,\aO) \in Q \times \AO$ since $|(\tr,\delta)(q,\aO,\AT)| = 1$ by base-case assumption and $t_{q,(\tr,\delta)(q,\aO,\AT)} \in (\tr,\delta)(q,\aO,\AT)$ by definition of a delayed response. Let $\sO_t$ be a \pO\/ winning strategy in $g_{q_0}(t)$, and let us define a \pO\/ strategy $\sO$ in $g_{q_0}$ by $\sO(\aHT) := \sO_t(|\aHT|)$ for all $\aHT \in \AT^*$. Let $\aHTinf \in \AT^\omega$. Invoking Lemma~\ref{lem:2p-t-1p} with $F(q,\aO) := \AT$ yields $\tr_t \circ \hOS_t(\sO_t,\omega) = \tr \circ \hOS(\sO,\aHTinf)$. Since $\aHTinf$ is arbitrary and $\sO_t$ is winning, $\sO$ is winning in $g_{q_0}$.

For the inductive case, let  $q_1 \in Q$ be a state involving \pT\/. For all delayed $q_1$-responses $\overline{e}$, fewer states involve \pT\/ in $g_{q_0}(q_1,\overline{e})$ than in $g_{q_0}$, by Observation~\ref{obs:cg-involve-less}.\ref{obs:cg-involve-less1}. By Observation~\ref{obs:cg-involve-less}.\ref{obs:cg-involve-less2}, for all delayed $q_1$-responses $\overline{e}$, every delayed responses $t$ in $g_{q_0}(q_1,\overline{e})$ is also a delayed response in $g_{q_0}$, and $t_{q_1} = \overline{e}$ and $g_{q_0}(q_1,\overline{e})(t) = g_{q_0}(t)$; by assumption \pO\/ wins $g_{q_0}(t)$ for all delayed responses $t$ in $g_{q_0}$, so \pO\/ wins $g_{q_0}(t) = g_{q_0}(q_1,\overline{e})(t)$ for all delayed responses $t$ in $g_{q_0}(q_1,\overline{e})$; by I.H. let $\sO^{\overline{e}}$ be a \pO\/ winning (finite memory) strategy in $g_{q_0}(q_1,\overline{e})$.

Let us make a case disjunction on whether some $\sO^{\overline{e}}$ avoids $q_1$ (starting from $q_0$), i.e. $\Delta_{q_0} \circ \hOS(\sO^{\overline{e}},\aHT) \neq q_1$ for all $\aHT \in \AT^*$. First case, there exists such an $\sO^{\overline{e}}$, which is also a \pO\/ winning (finite memory) strategy in $g_{q_0}$ as proved below. Recall that $g_{q_0}(q_1,\overline{e}) = \langle  \AO, \AT, Q, q_0, \delta', \Col, \tr', W\rangle$, where  for all $(\aO,\aT, q) \in \AO \times \AT \times (Q \setminus \{q_1\})$ we have $(\tr',\delta')(q,\aO,\aT) = (\tr,\delta)(q,\aO,\aT)$ and $(\tr',\delta')(q_1,\aO,\aT) = \overline{e}_{(\tr,\delta)(q_1,\aO,\AT)}$. Let us prove the following by induction on $\aHT$.
\begin{align}
\forall \aHT \in \AT^*,\,\Delta' \circ \hOS(\sO^{\overline{e}}, \aHT) = \Delta \circ \hOS(\sO^{\overline{e}}, \aHT) \label{claim:interleaving-prefix removal-eq1}
\end{align}
For the base case $\Delta' \circ \hOS(\sO^{\overline{e}}, \e) = q_0 = \Delta \circ \hOS(\sO^{\overline{e}}, \e)$. For the inductive case let $\aHT \in \AT^*$ and $\aT \in \AT$. Then
\begin{align*}
\Delta' \circ \hOS(\sO^{\overline{e}}, \aHT \aT) & = \delta'(\Delta' \circ \hOS(\sO^{\overline{e}}, \aHT), \sO^{\overline{e}}(\aHT),\aT) \mbox{ by definition of }\Delta',\\
	& =  \delta'(\Delta \circ \hOS(\sO^{\overline{e}}, \aHT), \sO^{\overline{e}}(\aHT),\aT) \mbox{ by I.H.,}\\
	& = \delta(\Delta \circ \hOS(\sO^{\overline{e}}, \aHT), \sO^{\overline{e}}(\aHT),\aT) \mbox{ by definition of }\delta'\mbox{ and since }\sO^{\overline{e}}\mbox{ avoids }q_1,\\
	& = \Delta \circ \hOS(\sO^{\overline{e}}, \aHT \aT) \mbox{ by definition of }\Delta\mbox{, thus completing the induction.}
\end{align*}
Let us now prove the following by induction on $\aHT$.
\begin{align}
\forall \aHT \in \AT^*,\,\tr' \circ \hOS(\sO^{\overline{e}}, \aHT) = \tr \circ \hOS(\sO^{\overline{e}}, \aHT) \label{claim:interleaving-prefix removal-eq2}
\end{align}
For the base case $\tr' \circ \hOS(\sO^{\overline{e}}, \e) = \e = \tr \circ \hOS(\sO^{\overline{e}}, \e)$. For the inductive case let $\aHT \in \AT^*$ and $\aT \in \AT$. Then
\begin{align*}
\tr' \circ \hOS(\sO^{\overline{e}}, \aHT \aT) & = \tr'(\hOS(\sO^{\overline{e}}, \aHT) \cdot (\sO^{\overline{e}}(\aHT), \aT)) \mbox{ by definition of }\hOS,\\
	& = \tr' \circ \hOS(\sO^{\overline{e}}, \aHT) \cdot \tr'(\Delta'\circ \hOS(\sO^{\overline{e}}, \aHT), \sO^{\overline{e}}(\aHT), \aT) \mbox{ by definition of }\tr',\\
	& = \tr \circ \hOS(\sO^{\overline{e}}, \aHT) \cdot \tr'(\Delta'\circ \hOS(\sO^{\overline{e}}, \aHT), \sO^{\overline{e}}(\aHT), \aT) \mbox{ by I.H.,}\\
	& = \tr \circ \hOS(\sO^{\overline{e}}, \aHT) \cdot \tr'(\Delta\circ \hOS(\sO^{\overline{e}}, \aHT), \sO^{\overline{e}}(\aHT), \aT) \mbox{ by Claim~(\ref{claim:interleaving-prefix removal-eq1}),}\\
	& = \tr \circ \hOS(\sO^{\overline{e}}, \aHT) \cdot \tr(\Delta\circ \hOS(\sO^{\overline{e}}, \aHT), \sO^{\overline{e}}(\aHT), \aT) \mbox{ by definition of }\tr' \mbox{ and since }\sO^{\overline{e}}\mbox{ avoids }q_1,\\
	& = \tr(\hOS(\sO^{\overline{e}}, \aHT) \cdot (\sO^{\overline{e}}(\aHT), \aT)) \mbox{ by definition of }\tr,\\
	& = \tr \circ \hOS(\sO^{\overline{e}}, \aHT \aT)\mbox{ by definition of }\hOS\mbox{, thus completing the induction.}
\end{align*}
Claim~(\ref{claim:interleaving-prefix removal-eq2}) is lifted to infinite argument: $\forall \aHTinf \in \AT^\omega,\,\tr' \circ \hOS(\sO^{\overline{e}}, \aHTinf) = \tr \circ \hOS(\sO^{\overline{e}}, \aHTinf)$. Since $\sO^{\overline{e}}$ wins $g_{q_0}(q_1,\overline{e})$, $\tr' \circ \hOS(\sO^{\overline{e}}, \AT^\omega) \subseteq W$, so  $\tr \circ \hOS(\sO^{\overline{e}}, \AT^\omega) \subseteq W$ and $\sO^{\overline{e}}$ wins $g_{q_0}$.

Second case, for all delayed $q_1$-responses $\overline{e}$, the strategy $\sO^{\overline{e}}$ does not avoid $q_1$, so let $\aHT^{\overline{e}} \in \AT^*$ be such that $\Delta_{q_0} \circ \hOS(\sO^{\overline{e}},\aHT^{\overline{e}}) = q_1$. For all $\overline{e}$, let $\sO^{\overline{e}}_{q_1}$ be the \pO\/ strategy defined by $\sO^{\overline{e}}_{q_1}(\aHT) := \sO^{\overline{e}}(\aHT^{\overline{e}} \aHT)$ for all $\aHT \in \AT^*$. Since $\sO^{\overline{e}}$ wins $g_{q_0}(q_1,\overline{e})$, by Lemma~\ref{lem:sub-strat-trace} $\sO^{\overline{e}}_{q_1}$ wins $g_{q_1}(q_1,\overline{e})$. Therefore there exists a \pO\/ winning (finite-memory) strategy in $g_{q_0}$, by invoking Lemma~\ref{lem:win-delayed-response-implies-win} with the $\sO^{\overline{e}}_{q_1}$ and one of the $\sO^{\overline{e}}$.
\end{proof}

\begin{proof}[Proof of Theorem~\ref{thm:interleaving-win-all-eq}]
\ref{thm:interleaving-win-all-eq1} $\Rightarrow$\ref{thm:interleaving-win-all-eq2} by Lemma~\ref{lem:win-implies-win-delayed-response}

\ref{thm:interleaving-win-all-eq2} $\Rightarrow$\ref{thm:interleaving-win-all-eq1}
Let $q_0 \in Q$. The proof proceeds by induction on the number of states that involve \pT\/. The base case is verbatim like the base case from the proof of Theorem~\ref{thm:interleaving-prefix removal-eq}. For the inductive case (similar but simpler), let  $q_1 \in Q$ be a state involving \pT\/. For all delayed $q_1$-responses $\overline{e}$, fewer states involve \pT\/ in $g_{q_0}(q_1,\overline{e})$ and $g_{q_1}(q_1,\overline{e})$ than in $g_{q_0}$, by Observation~\ref{obs:cg-involve-less}.\ref{obs:cg-involve-less1}. By Observation~\ref{obs:cg-involve-less}.\ref{obs:cg-involve-less2}, for all delayed $q_1$-responses $\overline{e}$, every delayed responses $t$ in $g_{q_i}(q_1,\overline{e})$, where $i \in \{0,1\}$, is also a delayed response in $g_{q_i}$, and $t_{q_1} = \overline{e}$ and $g_{q_i}(q_1,\overline{e})(t) = g_{q_i}(t)$; by assumption \pO\/ wins $g_{q_i}(t)$ for all delayed responses $t$ in $g_{q_i}$, so \pO\/ wins $g_{q_i}(t) = g_{q_i}(q_1,\overline{e})(t)$ for all delayed responses $t$ in $g_{q_i}(q_1,\overline{e})$; by I.H. let $\sO^{\overline{e}}_i$ be a \pO\/ winning (finite memory) strategy in $g_{q_i}(q_1,\overline{e})$. Therefore there exists a \pO\/ winning (finite-memory) strategy in $g_{q_0}$, by invoking Lemma~\ref{lem:win-delayed-response-implies-win} with the $\sO^{\overline{e}}_{1}$ and one of the $\sO^{\overline{e}}_0$.
\end{proof}

\begin{observation}
Consider a game $g = \langle  \AO, \AT, Q, q_0, \delta, \Col, \tr, W\rangle$ with finite $Q$ and $\Col$.
\begin{enumerate}
\item The $(\tr,\delta)(q,\aO,\AT)$ have cardinality at most $\Col \times Q$,  and there are at most $2^{|\Col||Q|}$ many of them, i.e. $k_q \leq 2^{|\Col||Q|}$.

\item There are at most $(|\Col||Q|)^{|Q|\left(2^{|\Col||Q|}\right)}$ \pT\/ delayed responses.
\end{enumerate}
\begin{proof}
\begin{enumerate}
\item The $(\tr,\delta)(q,\aO,\AT)$ are all included in $\Col \times Q$, which is finite like $Q$ and $\Col$.

\item Consider $\otimes_{q\in Q, i \leq k_q}E_i^q$: there are at most $|Q|\left(2^{|\Col||Q|}\right)$ pairs $(q,i)$ since $k_q \leq 2^{|\Col||Q|}$, as argued above, and each $E_i^q$ has cardinality at most $\Col \times Q$.
\end{enumerate}
\end{proof}
\end{observation}

\begin{proof}[Proof of Corollary~\ref{cor:interleaving-win-one-eq}]
By invoking the characterization from Theorem~\ref{thm:interleaving-prefix removal-eq} (resp. Theorem~\ref{thm:interleaving-win-all-eq}). For all $q \in Q$, the $(\tr,\delta)(q,\aO,\AT)$ are subsets of $\Col \times Q$, so they have cardinality at most $\Col \times Q$  and there are at most $2^{|\Col||Q|}$ many of them. So, for all $q \in Q$ there are at most $(|\Col||Q|)^{\left(2^{|\Col||Q|}\right)}$ delayed $q$-responses, and there are at most $(|\Col||Q|)^{|Q|\left(2^{|\Col||Q|}\right)}$ \pT\/ delayed responses. For each of them one should decide whether a one-player game with actions $\AO$, states $Q$, colors $C$, and winning condition $W$ is won by the player, which takes at most $f(|\AO|,|Q|,|\Col|)$ steps by assumption. Hence the term $f(|\AO|,|Q|,|\Col|) \cdot (|\Col||Q|)^{|Q|\left(2^{|\Col||Q|}\right)}$. The term $|Q||\AO||\AT|$ comes from some cleaning operation: for each state $q \in Q$, if $(\tr,\delta)(q,\aO,\AT) = (\tr,\delta)(q,\aO',\AT)$, then $\aO$ and $\aO'$ are ``equivalent'' enough.
\end{proof}

\section{Existence of \pT\/ almost surely winning random strategies}

\begin{proof}[Proof of \ref{thm:no-win-all-random-win1} $\Rightarrow$ \ref{thm:no-win-all-random-win3} from Theorem~\ref{thm:no-win-all-random-win}]
Let $p \in ]0,\frac{1}{|\AT|}]$ and let $\sT$ be a \pT\/ stochastic strategy that always assigns probability at least $p$ to every action. 

For all $q \in Q$, by contraposition of Theorem~\ref{thm:interleaving-prefix removal-eq} let $t_q$ be a delayed response (in $g_q$) such that \pO\/ loses the one-player game $g_q(t_q)$. For all $n \in \N$, anytime a play reaches the state $q$, the probability that from then on \pT\/ follows $t_q$ for $n$ rounds in a row, as if second-guessing \pO\/, is  greater than or equal to $p^n$.

Consider a play where \pT\/ follows $\sT$. Let $q$ be a state that is visited infinitely often. (Such a state exists since $Q$ is finite.) Thanks to the argument above, for all $n \in \N$, the probability that, at some point, \pT\/ follows $t_q$ for $n$ rounds in a row from $q$ on is one. Since the countable intersection of measure-one sets has also measure one, the probability that, for all $n \in \N$, at some point \pT\/ follows $t_q$ for $n$ rounds in a row from $q$ on is one.

Let $(\aH^n)_{n \in \N}$ be the corresponding full histories. Since $\AO$ and $\AT$ are finite, the tree induced by prefix closure of the $(\aH^n)_{n \in \N}$ is finitely branching, so by Koenig's Lemma it has an infinite path $\aHinf$, which corresponds to \pT\/ following $t_q$ infinitely many rounds in a row. So $\tr(\aHinf) \notin W$. By factor-prefix closeness the original play is also losing for \pO\/, i.e. winning for \pT\/.

\end{proof}

\section{The special case of stateless (i.e. one-state) games}

\begin{proof}[Proof of Observation~\ref{obs:factor-pref-set}]
Let $\Col$ be finite and $W \subseteq \Col^\omega$ be factor-prefix complete. Let $\Col_0 \subseteq \Col$ and let $\colHinf \in W$ have factors in $\Col_0^*$ of arbitrary length. So $\colHinf$ has factors in $\Col_0^*$ of arbitrary length and also occurring arbitrarily far in the tail. So by finiteness of $\Col_0$, for all $n \in \N$ there exists a word in $\Col_0^n$ that occurs as factor arbitrarily far in the tail of $\colHinf$. The set induced by the prefix closure of these words is a finitely-branching tree, so by Koenig’s Lemma it has an infinite path $\colHinf’ \in \Col_0^\omega$. By construction the prefixes of $\colHinf’$ occur as factors arbitrarily far in the tail of $\colHinf \in W$, so $\colHinf’ \in W$ by factor-prefix completeness. This $\colHinf’ \in W \cap \Col_0^\omega$ witnesses factor-set completeness.
\end{proof}

\begin{proof}[Proof of Theorem~\ref{thm:stateless-semi-rand-det}]
Let us assume that there is no \pO\/ winning strategy in the stateless game. Let $\Col_1,\dots,\Col_k$ be the elements of the set $\{\tr(\aO,\AT)\,\mid\,\aO \in \AO\}$. By contraposition of Corollary~\ref{cor:eq-win-abs} there exists $\overline{c} \in \Col_1 \times \dots \times \Col_k$ such that $W \cap \Col_0 = \emptyset$, where $\Col_0 := \{\overline{c}_1,\dots,\overline{c}_k\}$. By definition of the $\Col_i$, for all $\aO \in \AO$ there exists some $i$ such that $\tr(\aO,\AT) = \Col_i$, so there exists $\aT \in \AT$ such that $\tr(\aO,\aT) = \overline{\col}_i$. By skolemization, there exists $f: \AO \to \AT$ such that $\tr(\aO,f(\aO)) \in \Col_0$ for all $\aO \in \AO$. Then $\{f^{-1}(\aT)\,\mid\, \aT \in \AT\}$ form a partition of $\AO$.

Let the full-support Markovian strategy used by \pT\/ be defined by a probability distribution $p: \AT \to \R$, where $0 < p(\aT)$ for all $\aT \in \AT$. At an arbitrary stage of the interaction let $p': \AO \to \R$ be the probability distribution from which \pO\/ draws her next action $X$. Let $Y$ be the independent random variable for the action of \pT\/.
\begin{align*}
 \p(\tr(X,Y) \in \Col_0) & \geq \p(X \in f^{-1}(Y))\\
 	& = \sum_{\aO \in \AO}  \p(X = \aO \wedge Y =f(\aO))\\
 	& = \sum_{\aO \in \AO}  p'(\aO)p(\beta(\aO)) \mbox{ by independence of } X \mbox{ and }Y\\
 	& \geq \sum_{\aO \in \AO}  p'(\aO) \min_{\aT \in \AT}p(\aT)\\
 	& = \min_{\aT \in \AT}p(\aT)> 0 \mbox{ since }\sum_{\aO \in \AO}  p'(\aO) = 1
 \end{align*}
So, starting from any stage of the interaction, the probability that the procuded color is in $\Col_0$ for $n$ times in a row is at least $(\min_{\aT \in \AT}p(\aT))^n$. In particular it is positive and independent of the starting stage. Therefore, the probability that the procuded color is in $\Col_0$ $n$ times in a row somewhere in the run is $1$. Since a countable union of sets of measure zero also has measure zero, the probability that a run has factors in $\Col_0^*$ of arbitrary length is also $1$.

By the assumed factor-set completeness, if the run has factors in $\Col_0^*$ of arbitrary length the run is not in $W$. This shows that the probability that the run is in $W$ is 0.
\end{proof}

\section{Applications}

\begin{proof}[Proof of Lemma~\ref{lem:bool-combo}]
\begin{enumerate}
\item Let $(W_i)_{i \in I}$ be factor-set complete languages, let $W := \cup_{i \in I}W_i$, let $\Col_0$ be set, let $\aHinf\in W$ have factors of unbounded length over $\Col_0$. $\aHinf\in W_i$ for some $i \in I$, so $W_i \cap \Col_0^\omega \neq \emptyset$ by factor-set completeness of $W_i$, so $W \cap \Col_0^\omega \neq \emptyset$.

\item Let $(W_i)_{i \in I}$ be interleaving closed languages, and let $W := \cap_{i \in I}W_i$. Let $\aHinf,\aHinf'\in W$ and let $\aHinf''$ be obtained by interleaving of $\aHinf$ and $\aHinf$. For all $i \in I$, $\aHinf,\aHinf'\in W_i$, so $\aHinf'' \in W_i$ by interleving closeness of $W_i$. Therefore $\aHinf'' \in W$.

\item Let $(W_i)_{i \in I}$ be factor-prefix complete languages, let $W := \cap_{i \in I}W_i$ (resp. $W:= \cup_{i \in I}W_i$), let $\aHinf\in W$, and let the prefixes of some $\aHinf'$ occur arbitrarily far in the tail of $\aHinf$.  By factor-prefix completeness of each $W_i$, $\aHinf'$ belongs to each $W_i$, i.e. to $W$. (resp. $\aHinf\in W_i$ for some $i \in I$, so $\aHinf'\in W_i$, so by factor-prefix completeness of $W_i$, so $\aHinf'\in W$.)
\end{enumerate}
\end{proof}

\begin{proof}[Proof of lemma~\ref{lem:combine-closure}]
\begin{enumerate}
\item Let $W$ be closed under interleaving, let $\colH\colHinf, \colH
 \colHinf
 \in W$, and let $\colHinf^{(2)}$ be obtained by interleaving $\colHinf$ and $\colHinf
$. So $\colH\colH
\colHinf^{(2)}$ can be obtained by interleaving $\colH\colHinf$ and $\colH
\colHinf
$, so it is in $W$.

\item Let $W$ be factor-prefix complete and let $\colH \colHinf \in W$. Let the prefixes of some $\colHinf
$ occur arbitrarily far in the tail of $\colHinf$. These prefixes occur also arbitrarily far in the tail of $\colH \cdot \colHinf$, so $\colHinf
 \ni W$ by factor-prefix completeness.

Let $W$ be factor-set complete and let $\colH \colHinf \in W$. Let $\Col_0$ and let us assume that $\colHinf$ has factors over $\Col_0$ of arbitrary length. So does $\colH \colHinf$, so there exists $\colHinf
 \in W \cap \Col_0^\omega$ by factor-set completeness.

\item Similar as above.

\item Let $W$ be factor-prefix complete, and let $\colHinf$ be obtained by interleaving some $\colHinf^0,\colHinf^1 \in W$, and let the prefixes of some $\colHinf
$ occur infinitely far in the tail of $\colHinf$. Fix a concrete interleaving of $\colHinf^0$ and $\colHinf^1$ that yields $\colHinf$. If it consumes only a finite prefix of $\colHinf^i$, the prefixes of $\colHinf
$ occur infinitely far in the tail of $\colHinf^{1-i}$, and $\colHinf \in W$ by factor-prefix completeness. So, from now on let us assume that the fixed interleaving consumes both $\colHinf^0$ and $\colHinf^1$ entirely.

Let us pair $\colHinf$ and a binary sequence as follows: for all $n \in \N$ let $\rho_n := (\colHinf_n,i)$  if $\colHinf_n$ came from $\colHinf^i$ during the fixed interleaving process. By cardinality argument, for each prefix $\colH
$ of $\colHinf
$ there exists a binary word $u$ such that $(\colH
,u)$ occurs arbitrarily far in $\rho$. Then, by Koenig
s Lemma, there exists ${\bf u} \in \{0,1\}^\omega$ such that for all $n \in \N$ the word $(\colHinf
_{\leq n}, {\bf u}_{\leq n})$ occurs arbitrarily far in $\rho$. Consequently, for all $n \in \N$ the word $\colHinf
_{\leq n, {\bf u}}$ occurs arbitrarily far in $\colHinf^1$, where $\colHinf
_{\leq n, {\bf u}}$ is obtained by considering only the elements $\colHinf
_k$ such that ${\bf u}_k = 1$. The $\colHinf
_{\leq n, {\bf u}}$ are the prefixes of $\colHinf
_{{\bf u}}$, which is obtained similarly, so $\colHinf
_{{\bf u}} \in W$ by factor-prefix completeness. Likewise $\colHinf
_{1-{\bf u}} \in W$. By interleaving these two according to ${\bf u}$, one obtains $\colHinf$, which is therefore also in $W$.

The argument for factor-set completeness is similar.

\item Let $W$ be closed under prefix removal, and let $FP(W)$ be its factor-prefix completion. Let $\colH\colHinf \in FP(W)$. If $\colH\colHinf \in W$ then $\colHinf \in W \subseteq FP(W)$. If $\colH\colHinf \notin W$, let its prefixes occur arbitrarily far in some $\colHinf
 \in W$.The prefixes of $\colHinf$ also occur arbitrarily far in $\colHinf
$, so $\colHinf \in FP(W)$.
\end{enumerate}
\end{proof}

\begin{proof}[Proof sketch of Proposition~\ref{prop:interleaving-win-one-eq}]
The argument uses classical techniques. Let us first consider one-player parity games. If there is a winning run in such a game, there is one that is induced by a positional strategy, i.e. it ultimately goes along a simple cycle, and the minimal color along this cycle is even. For each reachable edge of even color, let us derive a graph by removing all the edges with lower color. If the derived graph has a cycle involving the special edge, there is a reachable cycle with an even minimal color in the original graph. So, the player wins iff one derived graph has a cycle. There are at most $m$ such derived graphs, where $m$ is the number of edges, and deciding the existence of cycles can be done in $O(m)$. So it is decidable in $O(m^2)$ whether a one-player parity game is winnable.

Let us now consider a Muller game $\langle  \AO, \AT, Q, q_0, \delta, \Col, \tr, W\rangle$. Let us expand it into a parity game by using the LAR datastructure sa in~\cite{Thomas97} . The number of states is now $|Q||\Col||\Col|!$ and the number of edges is bounded by $|\AO||\AT||\Col||\Col|!$. As is well-known, \pO\/ has a winning strategy in the Muller game iff she has one in the expanded parity game. So by Corollary~\ref{cor:interleaving-win-one-eq}, the decision can be made in big $O$ of
$$
2^{|Q||\Col||\Col|!}\left(|\AO||\AT| +  (|\AO||\AT||\Col||\Col|!)^2\cdot (|Q||\Col|^2|\Col|!)^{|Q||\Col||\Col|!\left(2^{|Q||\Col|^2|\Col|!}\right)}\right)
$$
\end{proof}

%Window parity $\Col := \N$. Let $\lambda \in \N$. The objective is that the least of any $\lambda$-consecutive members of the sequence is even: a sequence $(u_n)_{n \in \N} \in \R^\N$ is winning iff $\forall n,\, \min\{u_{n+i}\,\mid\, 0 \leq i \leq \lambda\} \in 2\N$.

\begin{proof}[Proof of Observation~\ref{obs:brl-fairness}]
\begin{enumerate}
\item Clear.
\item A delay $d$ witnessing $FF_\mathcal{C}(\gamma)$ is also a bound witnessing $BRL_\mathcal{C}(\colHinf)$.
\item  Let $\Col =\{0,1,2\}$, and let the problem $0$ (resp. $2$, resp. $1$) require the solution $1$ (resp. $2$, resp. $0$). Thus, $BRL_\mathcal{C}(02^\omega)$ since only there are always only two unsolved problem $0$ and $2$ (since a prior $2$ is solved by the next $2$). However, $\neg F_\mathcal{C}(02^\omega)$ since the problem $0$ is never solved. Conversely, $F_\mathcal{C}(01001\dots 10^n1\dots)$ since every $0$ is followed by a $1$, but $\neg BRL_\mathcal{C}(01001\dots 10^n1\dots)$ since they are arbitrarily long factor consisting of $0$.
\end{enumerate}
\end{proof}

\begin{proof}[Proof of Observation~\ref{obs:brl-richer-eq}]
$BRLD_\mathcal{C}(\colHinf)\,\Leftarrow\,BRL_\mathcal{C}(\colHinf)$ is clear by taking $d = 0$, and $BRLD_\mathcal{C}(\colHinf)\,\Rightarrow\,BRL_\mathcal{C}(\colHinf)$ is proved by taking $b + d$ for the new bound.
\end{proof}

\begin{proof}[Proof of Lemma~\ref{lem:brl-close-complete}]
\begin{itemize}
\item Prefix removal: more generally, a tail cannot have a bound worse than the sequence it comes from.

\item Interleaving: let $\colHinf,\colHinf' \in \Col^\omega$ be such that $BRL_\mathcal{C}(\colHinf)$ and $BRL_\mathcal{C}(\colHinf')$, and let $b$ and $b'$ be respective witnesses. Let $\colHinf''$ be obtained by interleaving of $\colHinf$ an $\colHinf'$. Then for all $n \in \N$, 
\begin{align*}
|\{k \in \N\,\mid\, k \leq n\,\wedge\,\neg S(k,n-k,\colHinf'')\}| & \leq |\{k \in \N\,\mid\, k \leq n\,\wedge\,\neg S(k,n-k,\colHinf')\}| \\
	& + |\{k \in \N\,\mid\, k \leq n\,\wedge\,\neg S(k,n-k,\colHinf')\}|
\end{align*}
So $b+b'$ witnesses $BRL_\mathcal{C}(\colHinf'')$.

\item factor-prefix: Let $\colHinf$ be such that $BRL_\mathcal{C}(\colHinf)$ and let the prefixes of some $\colHinf'$ occur arbitrarily far in the tail of $\colHinf$. Towards a contradiction, let us assume that $\neg BRL_\mathcal{C}(\colHinf')$. So for all $l \in \N$ there exists a finite prefix $\colH'$ of $\colHinf'$ such that $l$ problems have not yet being solved after $\colH'$ have just been read. But $\colH'$ is a factor of $\colHinf$ so there exists $\colH$ such that $\colH \colH' \sqsubseteq \colHinf$, and the number of unsolved problems after reading $\colH \colH'$ is at least as large as $l$. Since $l$ is arbitrary, this contradicts $BRL_\mathcal{C}(\colHinf)$. Therefore $BRL_\mathcal{C}(\colHinf')$.
\end{itemize}
\end{proof}

Example~\ref{ex:ic-not-comp} below shows that the interleaving-closeness may not be closed under complementation.

\begin{example}\label{ex:ic-not-comp}
\begin{enumerate}
\item The interleaving of two infinite sequences that are not eventually constant is not eventually constant, but interleaving the eventually constant sequences $0^\omega$ and $1^\omega$ may yield $(01)^\omega$, which is not eventually constant.

\item The interleaving of two bounded real-valued sequences with mean payoff in $[0,1]$ has also mean payoff in $[0,1]$. (The mean payoff of a bounded sequence $(x_n)_{n \in \N}$ is, e.g. $\liminf_{n \to \infty}\sum_{i = 0}^nx_i$.) However, interleaving $2^{\omega}$ and $(-2)^\omega$ may yield $(2 \cdot -2)^\omega$, which has mean payoff $0$.

\item Similarly, the interleaving of two sequences of bounded (in some metric space) partial sum has bounded partial sums. However, interleaving $1^{\omega}$ and $(-1)^\omega$, whose partial sums are $(n)_{n\ in \N}$ and  $(-n)_{n\ in \N}$, may yield $(1 \cdot -1)^\omega$, which has  bounded partial sums.
\end{enumerate}
\end{example}

Example~\ref{ex:ic-not-comp} below shows that interleaving-closeness may be closed under complementation.

\begin{example}\label{ex:ic-comp}
\begin{enumerate}
\item The first element of the sequence is in some fixed set.

\item A least one element of the sequence is in some fixed set. (Reachability.)

\item Assuming that the sequences are over a finite subsets of $\N$: the least element occurring infinitely many times is even. (Parity.)

\item Assuming that the sequences are real-valued: the partial sum has a lower bound. (Energy.)

\item Assuming that the sequences are real-valued: the sequence has positive meanpayoff.

\end{enumerate}
\end{example}

Interleaving closeness of the winning condition of a player will lead to a simple characterization of the existence of winning strategy for the player. First question, is interleaving closeness necessary for such a characterization. Second question, is there an informative characterization of

By Lemma~\ref{lem:bool-combo}.\ref{lem:bool-combo2}, if $W$, $\Col^\omega \setminus W$, $W'$, and $\Col^\omega \setminus W'$ are all interleaving-closed, so are their pairwise intersections and unions. A natural question is whether the union of two interleaving-closed sets can be interleaving-closed despite their complements not being so.

\begin{example}
Let $W$ be the subset of non-eventually constant sequences over some set $\Col$, and let $W'$ be a prefix-independant interleaving-closed set over $\Col$. Then $W \cup W'$ is also interleaving-closed and prefix-independent.

\begin{proof}
If the non-eventually constant sequence is fully ``consumed'', the interleaving is non-eventually constant. If the non-eventually constant sequence is not fully ``consumed'', the interleaving is in $W'$.
\end{proof}
\end{example}

\end{document}